\documentclass[12pt]{article}
\usepackage[utf8]{inputenc}
\usepackage{setspace}
\usepackage[plain]{fullpage}
\usepackage{amsmath,amsfonts,amssymb,amsthm}
\usepackage[authoryear, longnamesfirst]{natbib}
\usepackage{hyperref}
\usepackage{graphicx}
\usepackage{enumitem}
\usepackage{nicematrix}
\usepackage{xcolor}
\usepackage{algorithm}
\usepackage{thmtools}
\usepackage{algorithmic}
\usepackage{subcaption}
\usepackage{mathtools}
\usepackage{placeins}
\allowdisplaybreaks

\newtheorem{theorem}{Theorem}

\theoremstyle{remark}
\newtheorem{rem}{\protect\remarkname}
\providecommand{\remarkname}{Remark}

\title{Decision Theory for the Archetype Discovery Problem\thanks{We would like to thank Davide Viviano for sparking our interest in studying the archetype discovery problem. Montiel Olea gratefully acknowledges financial support by the National
Science Foundation Grant SES-2315600. }}
\author{Jos{\'e} Luis Montiel Olea \and Amilcar Velez \and Zhuoheng Xu \and Haomin Yu \and Shunqi Zhang\thanks{Department of Economics, Cornell University. Corresponding author: {{\href{mailto:sz672@cornell.edu}%
{sz672@cornell.edu}}}}} 
\date{May 17th, 2026}

\begin{document}

\onehalfspacing

\maketitle

\begin{abstract}
In the \emph{archetype discovery problem} a researcher wants to summarize $N$ heterogeneous policy effects of interest that vary over a discrete set of covariates. The goal is to partition the set of covariates into $K<N$ groups---the \emph{archetype sets}---and to provide a summary of the policy effects for each group. We use decision theory to show that, under a  weighted mean-squared-error criterion, a procedure analogous to the \emph{Sorted Group Average Treatment Effects} (GATES) solves the archetype discovery problem. The key difference is that, in the optimal procedure, archetype sets are obtained by weighted $K$-means clustering of the $N$ heterogeneous policy effects, instead of relying on $K$ equally-spaced quantiles. We show that the procedure that minimizes average risk for a given prior can be obtained by clustering the different values of the posterior mean estimate of the policy effects of interest. Similarly, an approximately minimax procedure in large samples can be obtained by clustering a consistent estimator of the policy effects. In both of these cases, an exact solution to the weighted $K$-means clustering problem can be found using a simple and well-known dynamic programming algorithm.  
\end{abstract}

\doublespacing

\section{Introduction}  

A researcher has access to a function $\phi$ that describes the variation of $N$ heterogeneous effects of a policy of interest over a discrete set of observable characteristics, $x \in \mathcal{X}$. The researcher would like to communicate this information to a policymaker that wishes to assess the impact of the policy over a group of heterogeneous individuals. Unfortunately, when $N$ is large, communicating $\phi$ to the policymaker could be challenging. The policymaker might prefer a simpler summary of these effects. For example, it is possible that the policymaker finds it simpler to interpret the \emph{Sorted Group Average Treatment Effects} (GATES) of \cite{Victor25}, where the policy effects in $\phi$ are first sorted to obtain a few quantiles and then $\phi$ is summarized by reporting averages of policy effects over the groups defined by such quantiles. The researcher thus faces a critical trade-off when communicating the policy effects: the reports should contain as much information as in the original function $\phi$ (that is, the reports should exhibit high \emph{fidelity}), but should also be easy to parse for the policymaker (the reports should be less \emph{complex} than the original function $\phi$).

In recent work, \cite{breza2025generalizability} develop a useful framework to think about how to navigate the aforementioned trade-off between high fidelity and low complexity. In particular, they consider the problem of how to partition the researcher's heterogeneous policy effects into $K$ groups of ``individual and environmental'' characteristics within which treatment effects are ``predictively stable.'' \cite{breza2025generalizability} refer to the groups of characteristics as \emph{archetype sets}, and to the problem of finding these groups as the problem of \emph{archetype discovery}. In an important departure from existing literature, they allow the researcher to pay a cost to \emph{admit ignorance} for some groups. 

In this paper, we apply statistical decision theory to the archetype discovery problem. We focus on the case in which the researcher is not allowed to admit ignorance and must summarize the policy effects for every observable characteristic. Since there are different recommendations in the literature regarding how to present and summarize heterogeneous policy effects, decision theory can guide their evaluation.\footnote{In addition to the GATES of \cite{Victor25} and the \emph{generalizability aware predictions} of \cite{breza2025generalizability}, other suggestions in the literature include the \emph{endogenous stratification} strategy studied by \cite{abadie2018endogenous}, the strategies based on decision trees in \cite{AtheyImbens2016,WagerAthey2018}, and the \emph{rashomon partitions} of \cite{venkateswaran2024robustly}.} Using statistical decision theory necessitates the specification of at least three ingredients: the menu of actions available to the researcher, the consequences of these actions as a function of a potentially unknown state of the world (i.e., the \emph{loss function}), and a statistical model which captures how the 
data distribution changes at each possible unknown state of the world. In terms of the action space, we focus on the case in which the researcher is only allowed to communicate a function $\bar{\phi}$ that takes at most $K$ values. This choice of action space ensures that the function reported by the researcher has low complexity relative to $\phi$, since $K$ will typically be small and, in particular, considerably smaller than the number of values that $\phi$ could take. Throughout the paper, we assume that the researcher takes $K$ as exogenously given. For the loss function, we assume that the main criterion used by the researcher in order to guide the construction of the report $\bar{\phi}$ is a typical weighted mean-squared error criterion. In particular, the loss associated to reporting the function $\bar{\phi}$ to summarize the original function $\phi$ takes the form:
\begin{equation}\label{eq:oracle_loss_noabst_intro}
L(\bar{\phi};\phi,p) 
\;\equiv\;
\sum_{x\in\mathcal{X}} p(x)\Big(\phi(x)-\bar{\phi}(x)\Big)^2,
\end{equation}
where $p(x)$ denotes a probability mass function over the discrete set $\mathcal{X}$ that is assumed to be known and chosen by the researcher. The loss function in \eqref{eq:oracle_loss_noabst_intro} encourages reports with high fidelity by penalizing the differences between $\phi$ and $\bar{\phi}$. 

Our first result concerns the \emph{oracle} choice of $\bar{\phi}$: that is, we are interested in characterizing the researcher's report that minimizes the loss function in \eqref{eq:oracle_loss_noabst_intro}, assuming that the policy effects in $\phi$ are known. Our Theorem \ref{thm:archetype discovery solution} shows that the oracle choice of $\bar{\phi}$ can be constructed by a procedure analogous to GATES. In the oracle procedure characterized by our theorem, the $N$ policy effects in $\phi$ are first sorted in increasing order and $K$ groups are then constructed by solving a weighted version of the popular $K$-means problem of \cite{mcqueen1967some} and \cite{lloyd1982least}. An important observation here is that the $K$ groups are based on clustering the $N$ scalar policy effects, not on the covariates, as in the recent work of \cite{kim2026causal}. It then follows that the main conceptual difference between the oracle solution and the popular GATES procedure is that in the oracle solution the groups are obtained by weighted $K$-means clustering of the policy effects, instead of relying on $K$ equally-spaced quantiles. Moreover, since the $K$ groups are defined by clustering the policy effects, the $k$-th archetype set in the oracle problem simply consists of all covariate values that map to a policy effect that belongs to the $k$-th cluster. The oracle report $\bar{\phi}$ is constant over each cluster, and the reported value is the cluster-conditional mean of $\phi$.

Before presenting the other results in the paper, it is worthwhile making two additional comments about the oracle choice of $\bar{\phi}$. First, a common critique of the GATES procedure (and of any other endogenous stratification procedure that creates groups based on outcomes) is that the resulting archetype sets could be \emph{``uninterpretable in terms of the original covariates, making it necessary to post-process the treatment effect distribution to gain insights about particular covariates''}; see the discussion in Section H.4 of \cite{venkateswaran2024robustly}. The proof of our Theorem \ref{thm:archetype discovery solution} shows that the interpretability (or lack thereof) of the oracle solution $\bar{\phi}$ crucially depends on the original policy effects $\phi$. It is perhaps convenient to explain this in the context of a simple example. Suppose that the original policy effects are described by a decision tree and that the covariate space is a finite subset of the unit square $[0,1]^2$. The different values of the policy effects partition the unit square in terms of axis-aligned rectangles. Interestingly, our results show that the archetype sets associated to the oracle solution will correspond to a \emph{coarsening} of this partition. This means that $\bar{\phi}$ inherits \emph{some} of the interpretability in $\phi$ (unions of axis-aligned rectangles), without having to impose interpretability as an additional constraint in the oracle problem. 

Second, since the $K$ groups in the oracle solution are obtained by solving a weighted $K$-means clustering problem of $N$ scalar policy effects (as opposed to solving a clustering problem involving a potentially high-dimensional vector of covariates), the exact solution to the oracle problem can be obtained in a computationally convenient way. It is known---see \cite{bruce1965optimum}, \cite{wu1989kn,wu1991optimal}, \cite{wang2011ckmeans}, \cite{gronlund2017fast}---that the weighted $K$-means clustering problem of $N$ sorted scalars can be solved in $O(KN^2)$ time using $O(KN)$ space by means of a simple dynamic programming algorithm that---for the sake of exposition---we present in Section \ref{subsec:dynamic_prog}. It is important to note that the key feature that enables the use of the dynamic programming algorithm is that the $K$ clusters of the policy effects are \emph{contiguous}: whenever two policy effects $\phi_n<\phi_{n'}$ are assigned to the same cluster, then any other policy effect that is ordered in between them will also be assigned to the same cluster. This means that the archetype sets in the oracle solution can be thought of as arising from a ``squashing'' of the level curves of the original function $\phi$.  

While our analysis of the oracle solution to the archetype discovery problem is helpful to understand how to summarize policy effects if one must, in practice, the function $\phi$ will need to be estimated from the data. Thus, we use statistical decision theory to understand if the variant of the GATES procedure that we have discussed so far is still a reasonable way of using sample data for archetype discovery. To this end, we first consider a Bayesian statistical decision problem, where we assume that the researcher observes a dataset $z \in \mathcal{Z}$ that is informative about $\phi$. The researcher has a (potentially nonparametric) statistical model $\{P_{\theta}\}_{\theta \in \Theta}$, where the index $\theta$ includes $\phi$, but potentially other nuisance parameters. We endow the researcher with a prior $\pi$ over $\Theta$. 

Our second result, Theorem \ref{thm:posterior_loss_minimization}, shows that---under the assumption that the model's parameter $\theta$ only enters the loss function through $\phi$---it is optimal to solve the archetype discovery problem by applying the dynamic programming algorithm to the posterior mean of $\phi$. This means that if the policy effects in $\phi$ are Conditional Average Treatment effects estimated from experimental or observational data,  one could use the Bayesian regression tree models for causal inference of \cite{hahn2020bayesian} based on the seminal work of \cite{hill2011bayesian}. The archetype sets can then be formed by clustering the posterior mean estimates of the conditional treatment effects. The resulting procedure is still analogous to GATES, but the key difference now is that the proxies for the treatment effects are based on a Bayesian posterior mean estimator, as opposed to a machine learning proxy. 

Our third result, Theorem \ref{thm:epsilon-minimax-archetypes}, dispenses the use of the prior $\pi$, and focuses on a researcher that is interested in minimizing worst-case risk. In order to analyze this problem, we need to be more explicit about the statistical model. To this end, we assume that the researcher has access to an estimator of the form
\begin{equation} \label{eq:estimator_intro} 
\hat\phi(x)=\phi(x)+\frac{\sigma_x}{\sqrt I}u_x,
\qquad
\{u_x\}_{x \in \mathcal{X}} \sim P,
\end{equation}
where we allow for distributions $P$ for which the marginals of $u_x$ have mean zero, variance one, and are subgaussian (with an optimal variance proxy of at most one). We treat $\sigma_x$ and $I$ as known. The hyperparameter $I$ plays a role analogous to the sample size: a large value of $I$ is interpreted as having a more precise estimator of $\phi(\cdot)$. We also note that we allow the error terms, $u_x$, to be correlated across the different values of the covariates $x \in \mathcal{X}$. Theorem \ref{thm:epsilon-minimax-archetypes} shows that applying the dynamic programming algorithm to the sorted values of $\hat{\phi}$ is approximately minimax (in a sense we make precise) provided $I$ is large enough. We also show in Theorem \ref{thm:minimax-rate} that the minimax \emph{regret} converges to zero at a rate that is at most of order $\sqrt{\log(|\mathcal X|)/I}$.

Finally, we consider the case in which---in addition to communicating the function $\bar{\phi}$ to the policymaker---the researcher can also provide a set of covariate values at which it is better to \emph{admit ignorance} or \emph{abstain}. We adopt the same loss function as in \cite{breza2025generalizability}, and just like in their framework, we consider a hyperparameter that measures how costly it is for the researcher to admit ignorance. We make a restriction on the type of abstention that the researcher can recommend to the policy maker. In particular, we require the abstention to \emph{respect} archetype sets: if the researcher recommends abstention for one covariate value in the $k$-th archetype set, then the researcher must necessarily recommend abstention for all the covariate values in such set. We think this is a reasonable restriction (consistent with the idea that the policymaker has some limited ability to parse complex functions). We show that this restriction also leads to reports $\bar{\phi}$ that cluster units based on their policy effects. If we further require the clusters to be contiguous (in a sense we make precise), we show there is a simple extension of our dynamic programming algorithm (that modifies the flow cost of the dynamic programming problem) that solves the archetype discovery problem with abstention.     

The rest of the paper is organized as follows. Section~\ref{section:basic_framework}
introduces the formal framework and characterizes the \emph{oracle} solution to the
archetype discovery problem (the best way of summarizing the information in $\phi$ when this function is known).  Section~\ref{section:unknown_phi} considers the case in
which the policy effects of interest are unknown and studies data-driven rules under
the average-risk and worst-case-risk criteria. Section~\ref{section:examples} presents
a simple numerical example to illustrate our results. 
Section~\ref{section:extension} discusses different extensions of the main results: in particular, we consider the archetype discovery problem with abstention, and also discusses alternative algorithms for solving the weighted $K$-means clustering problem. Section \ref{sec:application} revisits the application discussed in \cite{Victor25}: an experiment with the government of Haryana in North India designed to analyze the effects of a policy bundle that provided different incentives for immunization across several villages. We use this application to illustrate the differences between archetype sets constructed via weighted $K$-means clustering and archetype sets constructed using quantiles of the policy effects (as in the GATES procedure). Proofs
of the main results are collected in Appendix~\ref{app:main_results}. Additional results and supporting material are collected in Appendix ~\ref{app:support_materials}.

\section{Basic Framework} \label{section:basic_framework} 

\subsection{The archetype discovery problem}
A researcher has access to a mapping $\phi: \mathcal{X} \rightarrow \mathbb{R}$ that describes the effects of a policy of interest as a function of a discrete set of observable characteristics. The researcher would like to communicate the policy effects contained in $\phi$ to a policymaker, but the policymaker would prefer a simpler summary of these policy effects. 

Let $\{\phi_1, \ldots, \phi_N\}$ denote the $N \leq | \mathcal{X}|$ different values that the function $\phi$ takes. While the policymaker has a hard time parsing the information contained in $\phi$, the policymaker is comfortable processing and interpreting a function $\bar{\phi}: \mathcal{X} \rightarrow \mathbb{R}$ that takes only $K$ values, with $K\ll N$. The researcher is thus interested in finding a function $\bar{\phi}$ that provides a policy-relevant summary of the information in $\phi$. 

We refer to the set of characteristics associated to the $k$-th value of $\bar{\phi}$ as the $k$-th \emph{archetype set}. The problem of finding the archetype sets (and their corresponding values) was recently introduced by \cite{breza2025generalizability}. We follow their terminology and refer to this problem as the \emph{archetype discovery} problem. 

\subsection{Oracle solution to the archetype discovery problem} 

Let $p(x)$ denote a probability mass function over the discrete set $\mathcal{X}$. We follow \cite{breza2025generalizability} and assume that the main criterion used by the researcher in order to guide the construction of the summary $\bar{\phi}$ is a typical weighted squared loss. In particular, the loss of using $\bar{\phi}$ to summarize the original function $\phi$ takes the form:
\begin{equation}\label{eq:oracle_loss_noabst}
L(\bar{\phi};\phi,p) 
\;\equiv\;
\sum_{x\in\mathcal{X}} p(x)\Big(\phi(x)-\bar{\phi}(x)\Big)^2.
\end{equation}
Throughout, we assume, without loss of generality, that $p(x)>0$ for all $x \in \mathcal{X}$. We also assume that the researcher takes $K$---the number of values that the policymaker can comfortably parse---as given and tries to minimize \eqref{eq:oracle_loss_noabst} over all functions that take only $K$ values. More formally, let $\bar{\Phi}_{K}$ denote the set of all functions that map the set $\mathcal{X}$ to the real line and that take at most $K$ values; that is
\[ \bar{\Phi}_{K} \equiv \left\{ \bar{\phi}:\mathcal{X} \rightarrow \mathbb{R} \: \mid \: | \bar{\phi}(\mathcal{X})| \leq K   \right\}, \]
where $\bar{\phi}(\mathcal{X})$ denotes the image of the set $\mathcal{X}$ under the function $\bar{\phi}$. We say that $\bar{\phi}^*$ is an \emph{oracle solution} to the archetype discovery problem if 
\begin{equation} \label{eq:archetype discovery problem}
L(\bar{\phi}^*,\phi,p) = \inf_{\bar{\phi} \in \bar{\Phi}_{K}} L(\bar{\phi},\phi,p).
\end{equation}
Let $\bar{\phi}^*_{k}$ denote the $k$-th value of the oracle solution $\bar{\phi}^*$. Define the \emph{oracle $k$-th archetype set} associated to the oracle solution $\bar{\phi}^*$ as
\begin{equation} \label{eq:oracle_set}
\mathcal{A}^*_k \equiv \{ x \in \mathcal{X} \: \mid \: \bar{\phi}^*(x) = \bar{\phi}^*_k \}.
\end{equation}
Our first result shows that the oracle solution to the archetype discovery problem can be obtained by optimally grouping the $N$ values $\{ \phi_1, \ldots, \phi_N\}$ into $K$ \emph{clusters}. We will show that this can be done by solving the one-dimensional version of the classical $K$-means problem in \cite{mcqueen1967some} and \cite{lloyd1982least}. The signal processing literature refers to this problem as \emph{optimum $K:N$ quantization}; see \citet{wu1989kn,wu1991optimal} and the references therein. 

In order to establish our first result, assume (without loss of generality) that the values of $\phi$ have been sorted and that $\phi_1<  \ldots < \phi_{N}$. Define a $K$-\emph{clustering} of the values $\phi_1 < \ldots <\phi_{N}$ as a surjective function $c:\{1,\ldots, N\} \rightarrow \{1,\ldots, K\}$. We interpret $c(i)$ as the cluster assigned to the value $\phi_i$. The $k$-th cluster associated to the clustering function $c(\cdot)$ can then be denoted as $C_k \equiv \{\phi_i \in \{\phi_1, \ldots, \phi_N\} \: \mid \: c(i) = k \}$. 

Given a clustering function, let $\mu_{k}(c)$ denote the (conditional) mean of $\phi$ within cluster $k$; that is
\begin{equation} \label{eq:mus and pi} 
\mu_k(c) \equiv \frac{\sum_{\{i \: : \: c(i)=k\}} p_i \phi_i }{\sum_{\{i \: : \: c(i)=k\}} p_i}, \quad \textrm{where} \quad p_i \equiv \sum_{\{x\::\: \phi(x)=\phi_i\}}p(x). 
\end{equation}
Consider then the problem of clustering the values $\{\phi_1, \ldots, \phi_{N}\}$ into $K$ clusters in order to minimize weighted squared loss:
\begin{equation} \label{eq: k-means clustering}
\min_{c:\{1, \ldots, N\} \rightarrow \{1,\ldots, K\}} \sum_{k=1}^{K}  \left( \sum_{\{i\: : \: c(i)=k \}} p_i \: (\phi_i - \mu_k(c))^2 \right),
\end{equation} 
where the minimum is taken over all $K$-clustering functions. 
\begin{theorem} \label{thm:archetype discovery solution}
Let $i:\mathcal{X} \rightarrow \{1,\ldots, N\}$ be the function such that $\phi(x) = \phi_{i(x)}$. If $c^*$ solves the clustering problem in \eqref{eq: k-means clustering}, then the function 
\[ \bar{\phi}^*(x) = \mu_{c^*(i(x))}(c^*)   \]
is an oracle solution to the archetype discovery problem. Moreover, the $k$-th archetype set associated to the oracle solution $\bar\phi^*$ equals
\begin{eqnarray*}
\mathcal{A}^*_k &\equiv& \{ x \in \mathcal{X} \: \mid \: \bar{\phi}^*(x) = \bar{\phi}^*_k \} \\
&=& \{x \in \mathcal{X} \: | \: c^*(i(x))=k\}.
\end{eqnarray*}

\end{theorem}

\begin{proof}
See Appendix \ref{subsec:archetype discovery solution} 
\end{proof}

Theorem \ref{thm:archetype discovery solution} shows that the archetype discovery problem in \eqref{eq:archetype discovery problem} can be solved in the following way. The researcher first needs to sort the $N$ different values of $\phi$. Then, the researcher optimally clusters these values into $K$ groups---by solving \eqref{eq: k-means clustering}. Note that the characteristics $x$ do not enter this process: the clustering is done over $\{\phi_1, \ldots, \phi_N\}$ not over $\mathcal{X}$. The solution to the archetype discovery problem, however, is a $K$-valued function $\bar{\phi}^*$ defined over $\mathcal{X}$ and the theorem shows that the $\bar{\phi}^*$ can be constructed as follows. For a given $x \in \mathcal{X}$ we first retrieve the index $i \in \{1, \ldots, N\}$ associated with the value that the original function $\phi(x)$ is assigned to. For example, suppose that when evaluated at $\tilde{x}$, we have $\phi(\tilde{x})= \phi_3$. Then $i(\tilde{x})=3$. Once we know that $\phi(x)$ takes the value $\phi_{i(x)}$ we look for the cluster to which $i(x)$ was assigned. Under the optimal clustering, this value is $c^*(i(x))$. Our result says that the solution to the archetype discovery problem can be obtained directly from the $K$ optimal clusters of the set $\{\phi_1, \ldots, \phi_N\}$ by simply reporting the (conditional) cluster mean associated to $c^*(i(x))$. That is $\bar{\phi}^*(x) = \mu_{c^*(i(x))}(c^*)$. 

We note that the solution to the archetype discovery problem described in Theorem \ref{thm:archetype discovery solution} is analogous to the \emph{Sorted Group Average Treatment Effects} (GATES) procedure described in \cite{Victor25}, assuming away the (machine learning) estimation of the heterogeneous effects. As explained by \cite{Victor25} (see Comment 3.6, p. 1135), the groups created by GATES are \emph{``based upon actual predicted treatment effect".} This means that the ``heterogeneity groups'' created by GATES are groups induced by the ML proxy predictor (see p. 1128 in their paper). When there is no sampling uncertainty, this is analogous to create groups based on the values of the policy effects contained in $\phi$. 

The key difference between the oracle solution in Theorem \ref{thm:archetype discovery solution} and the GATES procedure, is that the policy effects in $\phi$ are obtained by weighted $K$-means clustering of the $N$ heterogeneous policy effects, instead of relying on $K$ equally-spaced quantiles.\footnote{For instance, in their main application \cite{Victor25} estimate the GATES by quintiles of the ML proxies. See their Figure 4 in Section 6.2.}

\begin{rem}[\emph{The archetype sets ``squash'' the level curves of $\phi$}] \label{rem:measurability} 
For $i=1 \ldots N$, define the $i$-th \emph{level curve} of the function $\phi$ 
\[ G^{\phi}_i \equiv \phi^{-1}(\phi_i) = \{x \in \mathcal{X} \: | \: \phi(x)=\phi_{i}\}.  \]
Note that the $i$-th level curve $G^{\phi}_i$ collects the values of $x \in \mathcal{X}$ that satisfy $\phi(x)=\phi_i$. Note that the collection of level curves $G^{\phi} \equiv \{G_{1}^{\phi}, \ldots, G^{\phi}_N\}$ forms a partition of $\mathcal{X}$. Mathematically, this is the same as saying that $\phi$ induces a partition of the set of covariates $\mathcal{X}$ through its associated level curves. For illustration, suppose that there are two covariates $(x_1,x_2)$ that belong to the unit square $[0,1]^2$. Suppose that the function $\phi$ takes six values. Then, the function $\phi$ partitions the unit square into the six groups corresponding to the values of $(x_1,x_2)$ that evaluate to each of the values $\phi_i$. This is illustrated by Figure \ref{fig:running-example-thm1}.

Let $A^*_K$ be the set of all functions $a:\{1,\ldots, N\} \rightarrow \mathbb{R}$ such that $|a(\{1,\ldots, N\})| = K$ and consider then the set of functions 
\begin{equation} \label{eq:measurable_w.r.t_G}  
\bar{\Phi}^*_K(A_K,G^{\phi}) \equiv \left\{ \bar{\phi} \in \bar{\Phi}_{K} \: \Big | \: \bar{\phi}(x) = \sum_{i=1}^{N} a(i) \mathbf{1}\{x \in G^\phi_i\} \:  \textrm{for some} \: a \in A^*_K   \right \}.  
\end{equation}
Note that, by definition, any function in \eqref{eq:measurable_w.r.t_G} induces a \emph{coarser} partition over the set $\mathcal{X}$ than $\phi$. This means that if two vectors $x$ and $x'$ were in the same group $G^{\phi}_{i}$, then any function $\bar{\phi}$ belongs to \eqref{eq:measurable_w.r.t_G} must satisfy $\bar{\phi}(x)=\bar{\phi}(x')$.\footnote{More formally, it can be shown that the set $\bar{\Phi}^*_K(A_K,G^{\phi})$---which is a strict subset of $\Bar{\phi}_{K}$ coincides with the set of all functions that are measurable with respect to the $\sigma$-algebra generated by the partition $G^{\phi}$ and that take $K$ values.} Figure \ref{fig:running-example-thm1-measurable} gives an example of a partition that is coarser than the one induced by the function $\phi$ taking only six values.  Figure \ref{fig:running-example-thm1-nonmeasurable} gives an example of a partition that is not coarser, which means it cannot be generated by a function $\bar{\phi}$ in \eqref{eq:measurable_w.r.t_G}.

Note that Theorem \ref{thm:archetype discovery solution} implicitly shows that there exists an oracle solution of \eqref{eq:archetype discovery problem} is an element of $\bar{\Phi}^*_K(A_K,G^{\phi})$. Moreover, since it is known that the solution of the clustering problem in \eqref{eq: k-means clustering} produces \emph{contiguous} clusters (in the sense that if $\phi_n<\phi_{n'}$ are assigned to same cluster so will any intermediate value), then the archetype sets in the oracle solution will be given by a  ``squashing'' of the original ``level curves'' of the problem. This also means that the resulting archetype sets will inherit any interpretability associated to the original level curves of $\phi$. \qed 
\end{rem}

Finally, it is worthwhile to give a brief interpretation of our Theorem \ref{thm:archetype discovery solution} in terms of the well-known signal processing literature on optimum quantization; see \cite{wu1989kn,wu1991optimal}. Using the terminology from this literature, our Theorem \ref{thm:archetype discovery solution} says that the optimal low-complexity report of $\phi$ is the optimum $K$-point scalar quantizer of the push-forward distribution of $X \sim p(\cdot)$ by $\phi(\cdot)$, with archetype sets given by inverse images of quantizer cells.

\begin{figure}[t!]
    \centering
    \begin{subfigure}[t]{0.45\textwidth}
        \centering
        \includegraphics[width=\textwidth]{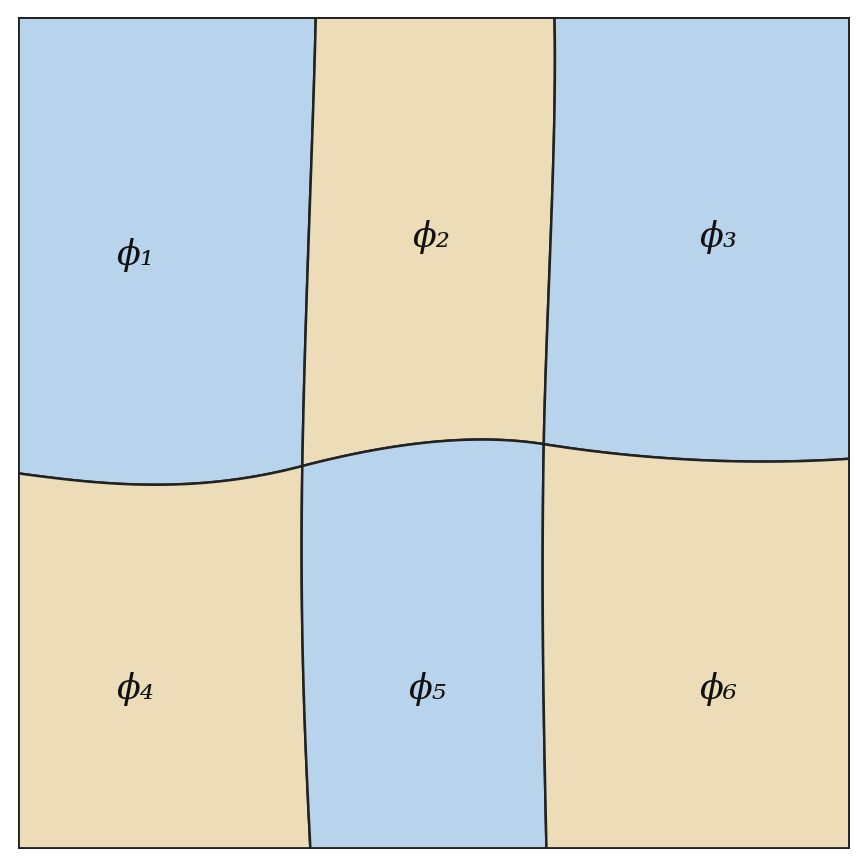}
        \caption{Coarser partition than the one induced by $\phi$.}
        \label{fig:running-example-thm1-measurable}
    \end{subfigure}
    \hfill
    \begin{subfigure}[t]{0.45\textwidth}
        \centering
        \includegraphics[width=\textwidth]{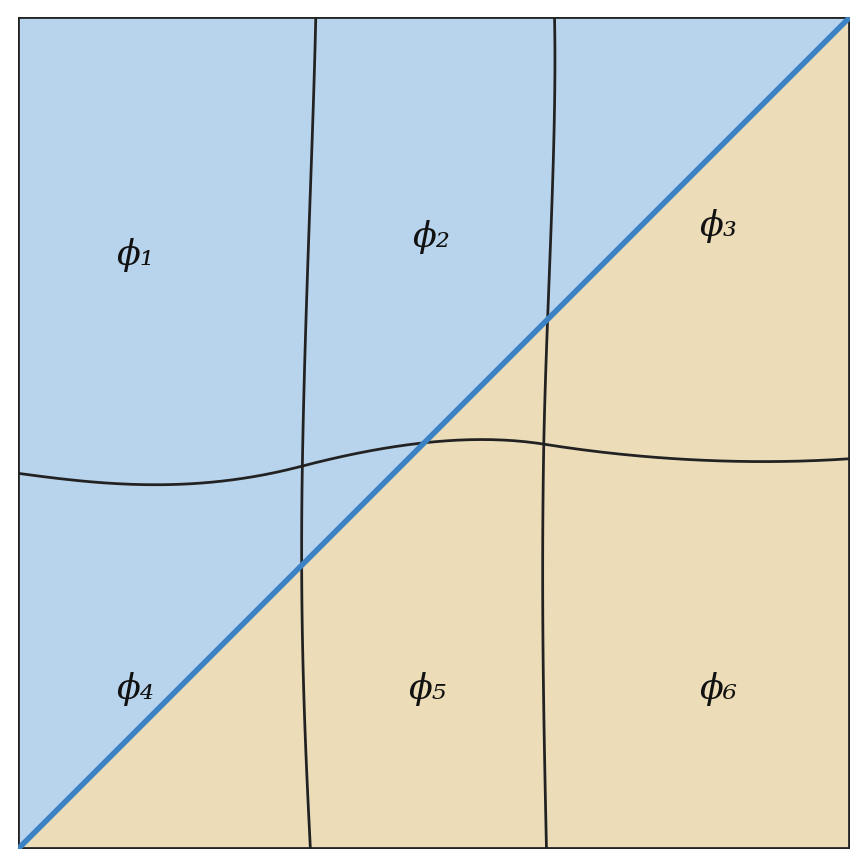}
        \caption{A partition that is not coarser than the one induced by $\phi$.}
        \label{fig:running-example-thm1-nonmeasurable}
    \end{subfigure}
    \caption{} 
    \label{fig:running-example-thm1}
\end{figure}

\subsection{Implementation of the oracle solution to the archetype discovery problem using dynamic programming} \label{subsec:dynamic_prog}

It is known that an exact solution to the one-dimensional clustering problem in \eqref{eq: k-means clustering} can be found using dynamic programming. The algorithm we present below is based on the seminal work on optimum quantization of \cite{bruce1965optimum}. We provide matlab and python scripts that implement such algorithm.\footnote{\cite{wang2011ckmeans} provide an implementation in R.} \cite{bruce1965optimum}'s algorithm performs well in our applications (and runs in a fraction of a second), but we note that there are other algorithms for finding an exact solution to the one-dimensional clustering problem in \eqref{eq: k-means clustering} that can improve both the runtime and the space requirements of the dynamic programming suggested in this paper. See \cite{gronlund2017fast} and the references therein. We decided to focus on \cite{bruce1965optimum}'s algorithm because of its simplicity.   

\emph{Dynamic Programming algorithm.} We are interested in grouping $\phi_1 < \ldots  <\phi_{N}$ into $K$ clusters using the weighted $K$-means objective function in \eqref{eq: k-means clustering}. The first key observation is that we can focus on clustering functions $c:\{1,\ldots, N\} \rightarrow \{1,\ldots, K\}$ that produce  \emph{contiguous} clusters: if $\phi_n < \phi_{n'}< \phi_{n{''}}$ and $c(n)=c(n'')=k$, then $c(n')=k$.  Note that any such clustering function can be characterized by a collection of $K-1$ ordered integers $q_0 \equiv 0 < q_1 < \ldots < q_{K-1}<q_{K} \equiv N$ where $c(i)=k$ if and only if $q_{k-1} < i \leq q_{k}$. The $k$-th cluster is thus given by $\{ \phi_i \: \mid \: q_{k-1} < i \leq q_{k}  \}$. Setting $q_0 \equiv 0$ and $q_K \equiv N$ guarantees that the $K$ clusters are well defined. 

Note that cardinality of the set of all contiguous clusters is given by the combination formula of $(N-1)$ choose $(K-1)$. A naive algorithm that works via enumeration becomes intractable for modest case and $K \ll N$ and $N$ is large. For example, if $N=100$ and $K=6$, the total number of elements is 715,231,440.

A key observation that allows for a dynamic programming algorithm to find the optimal clustering scheme is the following. Let $\widehat{q}^K_{N} = (q_1, \ldots, q_{K-1})$ be the vector that contains the $K-1$ ordered thresholds that define the optimal $K$ clusters. Lemma 1 in \cite{wu1989kn} shows that in the case of multiplicity of solutions, we can always select $\widehat{q}^K_{N}$ to be the solution with the largest thresholds (the set of solutions is totally ordered under the usual vector ordering). Let $q_{t}$ denote the $t$-th ($\leq K-1$) entry of the vector $\widehat{q}^K_{N}$. As discussed above the threshold $q_t$ defines the $t$-th cluster for $\phi_1< \ldots< \phi_N$. Consider the points $\phi_1< \ldots< \phi_{q_t}$. Lemma 2 in \cite{wu1989kn} shows that $(q_1, q_2, \ldots, q_{t-1})$ still characterize the optimal way of clustering the $q_t$ points $\phi_1< \ldots< \phi_{q_t})$ into $t$ clusters. 

This means that the clustering problem admits a recursive representation. For any $1 \leq k \leq K$ and $k \leq n \leq N$, let $D(k,n)$ denote the weighted mean-squared error of optimally grouping $\phi_1< \ldots< \phi_n)$ into $k$ clusters. For $n \leq n' \leq N$, let $C(n,n')$ be the squared-error of assigning points $\phi_{n} \leq \ldots \leq \phi_{n'}$ into one cluster. That is 
\[ C(n,n') \equiv \sum_{i=n}^{n'} p_i \left( \phi_i-\mu_{n,n'} \right)^2, \quad \textrm{ where } \mu_{n,n'} \equiv \frac{\sum_{i=n}^{n'} p_i \phi_i}{\sum_{i=n}^{n'} p_i}. \]
Note that Lemma 2 in \cite{wu1989kn} implies
\[ D(k,n) = \min_{a \in \{k,\ldots, n\}}  C(a,n) + D(k-1,a-1). \]
This defines a dynamic programming problem in the sense of \cite{bradley1977applied}. The \emph{state} of the problem is a vector $(k,n)$ where $k \leq n$, $1 \leq k \leq K$, and $k \leq n \leq N$. The state can be interpreted as making reference to the $n$ points $\phi_1< \ldots< \phi_n$ that have not been clustered and that need to be grouped into $k$ clusters. The action (or decision) in each state, denoted by $a \in \{k,\ldots, n\}$, refers to the points $\{\phi_a , \ldots, \phi_n\}$ that will be assigned to the right-most cluster. The \emph{flow} payoff (or more precisely, the flow \emph{cost}) is simply $C(a,n)$. The transition function maps a state $(k,n)$ and an action $a$ into the new state $(k-1,a-1)$, where we need to cluster $\{\phi_1, \ldots, \phi_{a-1}\}$ into $k-1$ clusters. 

It is known that there are other slightly more complicated algorithms that can improve the time and space requirements of the baseline dynamic programming algorithm; see Section 1.2 in \cite{gronlund2017fast}. In the numerical exercises where $N$ is in the thousands and $K$ is less than or equal to $10$, the runtime of the dynamic programming algorithm we have just discussed is almost a third of a second in a personal laptop, and therefore we recommend this procedure. 

It is also useful to contrast this exact one-dimensional solution of the $k$-means clustering problem with approximate algorithms commonly used for multidimensional $k$-means problems. We note that in our framework there is no need to use heuristic algorithms to solve the $K$-means problem. The popular iterative procedures of \cite{lloyd1982least}-\cite{forgy1965cluster}, and \cite{hartigan1979algorithm} have been applied to other clustering problems in econometrics; e.g., \cite{bonhomme2015grouped} and \cite{bonhomme2019distributional}. However, as noted by \cite{gronlund2017fast}, even for clustering problems  in one dimension, \emph{``Lloyd’s algorithm does not necessarily compute the optimal clustering, it is only a heuristic.}''

\section{Archetype Discovery problem when $\phi$ is unknown}  
\label{section:unknown_phi}
\subsection{Statistical Decision Theory for archetype discovery} 

In the previous section we described an oracle (or population) version of the archetype discovery problem, where we assumed that the researcher knew the function $\phi: \mathcal{X} \rightarrow \mathbb{R}$ describing the effects of the policy of interest. In this section we consider the case in which $\phi$ is potentially unknown to the researcher, but we assume there is a dataset available that is informative about this parameter. The question of interest is how to use the data to provide a summary $\bar{\phi}(\cdot)$ to the policymaker.  

Suppose that the researcher observes a dataset $z \in \mathcal{Z}$. As usual, we view the data as the realization of a $\mathcal{Z}$-valued random variable. The researcher has a statistical model $\{P_{\theta}\}_{\theta \in \Theta}$, where the parameter $\theta$ includes $\phi$ and potentially other nuisance parameters. But throughout this section, we maintain the assumption that the weights $p(\cdot)$ are known. These weights encode the populations or policy weights according to which approximation errors are evaluated. 

Just as we did before, we assume that the researcher takes $K$---the number of values that the policymaker can comfortably parse---as given. The main goal of the researcher is to provide a data-driven summary of the function $\phi$. The loss function that the researcher uses takes the form
\[ L(\bar{\phi},\theta) = L(\bar{\phi};\phi,p) 
\;\equiv\;
\sum_{x\in\mathcal{X}} p(x)\Big(\phi(x)-\bar{\phi}(x)\Big)^2. \]
This means that we are assuming that the model's parameter $\theta$ only enters the loss function through $\phi$. We denote any data-driven summary as a function $d: \mathcal{Z} \rightarrow \bar{\Phi}_{K}$. We use the standard terminology in statistical decision theory---see, for example, \cite{Ferguson67}---and refer to $d$ as the researcher's \emph{decision rule}. The risk of a decision rule at $\theta$ is defined as
\begin{equation} \label{eq:risk_function} 
R(d,\theta) \equiv \mathbb{E}_{Z \sim P_{\theta}} [ L(d(Z),\phi,p)].
\end{equation}

\subsection{Archetype discovery by minimizing average risk}  
Consider first the case where the researcher has a prior distribution $\pi$ over the parameter space $\Theta$. And consider the problem of finding a decision rule that minimizes average risk: $\mathbb{E}_{\theta \sim \pi}[R(d,\theta)]$. It is known that such a decision rule can be found by minimizing posterior loss (c.f., \cite{berger1985statistical}, Result 1, p. 159); that is, for each $z \in \mathcal{Z}$, set $d(z)$ to be any summary function $\bar{\phi}$ that solves the minimization problem
\begin{equation} \label{eq:posterior loss minimization}
\inf_{\bar{\phi} \in \bar{\Phi}_{K} } \mathbb{E}_{\phi \sim \pi} [L(\bar{\phi},\phi,p) \: | \: z  ].  
\end{equation}
Our next result shows that the posterior loss minimization problem in  \eqref{eq:posterior loss minimization} can be solved by applying the dynamic programming algorithm to the posterior mean of $\phi$. 

\begin{theorem} \label{thm:posterior_loss_minimization}
Any decision rule $d^*: \mathcal{Z} \rightarrow \bar{\Phi}_{K}$ for which $d^*(z)$ is a solution to the problem
\[ \inf_{\bar{\phi} \in \bar{\Phi}_{K}} L( \bar{\phi} , \mathbb{E}_{\phi \sim \pi}[\phi(x) \: | \: z] , p  ), \]
minimizes average risk with respect to $\pi$. 
\end{theorem}

\begin{proof}
See Appendix \ref{subsec:posterior_loss_minimizaation}
\end{proof}
This simple result shows that in order to find the decision rule $d^*$ that minimizes average risk it suffices to solve an \emph{oracle} archetype discovery problem where we pretend that the posterior mean function, $\mathbb{E}_{\phi \sim \pi}[\phi(x) \: | \: z]$, is the true function $\phi$. The results in the previous section then imply that if we define
\[ \widehat{\phi}(x) \equiv \mathbb{E}_{\phi \sim \pi}[\phi(x) \: | \: z]  \]
and sort its image as $\widehat{\phi}_1< \ldots < \widehat{\phi}_N$ (where $N \leq |\mathcal{X}|$), we can use the dynamic programming algorithm to optimally group $\widehat{\phi}_1< \ldots < \widehat{\phi}_N$ in $K$ clusters.

We think this result is interesting for two reasons. First, Theorem \ref{thm:posterior_loss_minimization} provides a simple (and principled) approach to deal with the fact that, in applications, the function $\phi$ will typically be unknown. Second, the proof of Theorem \ref{thm:archetype discovery solution} shows that $d^*(z)$ will be measurable with respect to the $\sigma$-algebra generated by $\widehat{\phi}$. This means that if the posterior mean function has a special structure in its domain (for example, if $\widehat{\phi}$ is a decision tree over $\mathcal{X}$), then archetype sets corresponding to $d^*(z)$ will be a coarsening over the partitions in $\mathcal{X}$ induced by the function $\widehat{\phi}$.

\subsection{$\epsilon$-minimax solution of the archetype discovery problem}  \label{subsection:epsilon-minimax} 

Consider now the problem of finding the decision rule that minimizes worst-case risk. That is, we are interested in solving the problem
\[ \adjustlimits \inf_{d} \sup_{\theta \in \Theta} R(d,\theta),   \]
where the risk function is defined in \eqref{eq:risk_function}, and $\theta$---which includes $\phi$---will denote the parameter of the statistical model used for the available data.

In order to make progress with the minimax problem, we focus on the case in which the available data consist of an estimator of the unknown function $\phi$. More precisely, suppose that we have an estimator of the form
\begin{equation} \label{eq:estimator} 
\hat\phi(x)=\phi(x)+\frac{\sigma_x}{\sqrt I}u_x,
\qquad
\{u_x\}_{x \in \mathcal{X}} \sim P. 
\end{equation}
We define the parameters of this statistical model to be $\theta \equiv (\phi,P)$ and we treat $\sigma_x$ and $I$ as known. The hyperparameter $I$ plays a role analogous to the sample size: a large value of $I$ is interpreted as having a more precise estimator of $\phi(\cdot)$. We also note that we are allowing for the error terms, $u_x$, to be correlated across the different values of the covariates $x \in \mathcal{X}$. 

\begin{rem} [Consistency of $\widehat{\phi}$]
It is important to note that the key aspect of the statistical model used in this section is the presence of a \emph{consistent} estimator of the policy effects, and not the rate at which the policy effects are estimated. While we have decided to work with a standard parametric rate, the results in this section go through replacing $\sqrt{I}$ by an arbitrary rate $r_{I}$ that diverges to infinity as $I$ grows large. The statistical model that we wanted to capture with \eqref{eq:estimator} was a simple model with discrete covariates where $P$ is a multivariate normal distribution with independent components. While a model like this might be plausible under some assumptions,  we remind the reader that, as noted by \cite{Victor25} in the case in which the policy effects are conditional average treatment effects, generic estimators cannot be regarded as consistent, unless further assumptions are made.  \qed
\end{rem} 

Fix an arbitrary constant $B>0$ and define the parameter space $\Theta$ to be any subset of
\begin{eqnarray} \label{eq: parameter space}
\Theta(B) \equiv \{ (\phi,P) \: &|& \: \textrm{for all } x \in \mathcal{X} \textrm{ and } t\in \mathbb{R} \quad  \phi(x) \in [-B,B] \nonumber  \\
&& \mathbb{E}_{P}[u_x]=0, \mathbb{E}[u^2_x]=1, \:  \mathbb{E}_{P}[\exp(t u_x)] \leq \exp(t^2/2) \}. 
\end{eqnarray}
The set $\Theta(B)$ only contains functions $\phi$ that are bounded in absolute value by $B$. In addition, the set $\Theta(B)$ only allows for distributions $P$ for which the marginals of $u_x$ have mean zero, variance one, and are subgaussian; e.g. see definition in \cite{rigollet2015}, \cite{Vershynin2018}, \cite{rivasplata2012} (with an optimal variance proxy of at most one). 

We are not aware of how to find the minimax rule when $\Theta = \Theta(B)$, or even when $\Theta$ is a strict subset of $\Theta(B)$ that only allows for error terms that are independent standard Gaussian random variables. Instead of insisting in finding an exact solution, we search for a solution that \emph{approximately solves} the minimax problem when $I$ is large enough.  

To this end, let $\bar{\Phi}_K(B) \equiv \{ \bar{\phi} \in \bar{\Phi}_K \: | \: \bar{\phi}(x) \in [-B,B]\}$. Define the plug-in decision rule, $d_{\textrm{plug-in}}$, to be any decision rule such that
\begin{equation} 
d_{\textrm{plug-in}}(\widehat{\phi})
\in
\arg\min_{\bar\phi\in\bar\Phi_K(B)}
L(\bar\phi;\hat\phi,p).
\label{eq:plugin-rule-minimax}
\end{equation}
Let $\mathcal{D}(B)$ be the set of all decision rules that map $\widehat{\phi}$ to $\bar{\Phi}_{K}(B)$. 
The minimax value of interest---as a function of $I$ and the parameter space $\Theta \subseteq \Theta(B)$---is
\begin{equation}
V(I,\Theta)
\equiv
\adjustlimits \inf_{d \in \mathcal{D}(B)} \sup_{\theta \in \Theta}R(d,\theta), \quad R(d,\theta)\equiv \mathbb{E}_{\widehat{\phi} \sim (\phi,P)}\left [L(d(\widehat{\phi});\phi,p)\right].
\label{eq:minimax-value-archetypes}
\end{equation}
denote the finite-sample minimax benchmark. The following theorem shows that the plug-in rule is \(\varepsilon\)-minimax (in the sense of \cite{Ferguson67}, Chapter 1, Definition 4, p.33) for all sufficiently large values of \(I\).

\begin{theorem}[$\varepsilon$-minimaxity of the plug-in rule]
\label{thm:epsilon-minimax-archetypes}
Fix $B$ and $\{\sigma_{x}\}_{x \in \mathcal{X}}$. Suppose that $\widehat{\phi}$ is generated according to the statistical model \eqref{eq:estimator}. Let $\Theta$ be an arbitrary nonempty subset of $\Theta(B)$. If an exact minimax rule exists for large enough $I$, then for every \(\varepsilon>0\), there exists \(I(\varepsilon)\) such that for all \(I\ge I(\varepsilon)\),
\[
V(I,\Theta) \le \sup_{\theta \in \Theta}R(d_{\textrm{plug-in}},\theta)\le V(I,\Theta) +\varepsilon.
\]
\end{theorem}

\begin{proof}
See Appendix \ref{app:proof-epsilon-minimax-theorem}.
\end{proof}

\subsection{Minimax regret in the archetype discovery problem}
We finalize the decision-theoretic analysis of the archetype discovery problem by analyzing the worst-case regret criterion. We define the regret loss of an action $\bar{\phi} \in \bar{\Phi}_{K}(B)$ to be: 
\begin{equation}\label{eq:regret_loss} 
\mathcal{L}(\bar{\phi};\phi,p) \equiv L(\bar\phi;\phi,p) - \inf_{\bar\phi \in \bar\Phi_K(B)} L(\bar\phi;\phi,p). 
\end{equation}
As usual, the regret loss above captures the excess loss of an action relative to the oracle solution of the archetype discovery problem.   
\begin{theorem}[Minimax Regret Rate]
\label{thm:minimax-rate}
Assume the statistical model in \eqref{eq:estimator} with parameter space $\Theta(B)$ given by \eqref{eq: parameter space}. Let $\bar\sigma := \sup_{x \in X}\sigma_x < \infty$. For all sufficiently large $I$,
\begin{equation} \label{eqn:minimax_regret_rate_thm_equation} 
\adjustlimits \inf_{d \in D(B)}
\sup_{\theta \in \Theta(B)}
\mathbb{E}_{\widehat{\phi} \sim (\phi,P)} \left[ \mathcal{L}(d(\widehat{\phi}); \phi,p) \right]
\;\le\;
8B \bar\sigma\, \sqrt{\frac{2 \log(2|\mathcal{X}|)}{I}}. 
\end{equation}
\end{theorem}

\begin{proof}
See Appendix \ref{app:proof-rate-epsilon-minimax-theorem}.
\end{proof}
Theorem \ref{thm:minimax-rate} shows that---even if the cardinality of $\mathcal{X}$ grows as a function of $I$---the worst-case regret of the optimal decision rule in the archetype discovery problem vanishes to zero as $I$ grows large (provided $|\mathcal{X}|$ does not grow too quickly and $\bar{\sigma}$ remains bounded). Thus, our theorem provides an upper bound on the \emph{minimax regret rate} associated to the archetype discovery problem. 

Theorem \ref{thm:minimax-rate} is established by showing that the worst-case regret of the \emph{plug-in} rule is bounded by the right-hand side of \eqref{eqn:minimax_regret_rate_thm_equation}. That is,
\[\sup_{\theta \in \Theta(B)}
\mathbb{E}_{\widehat{\phi} \sim (\phi,P)} \left[ \mathcal{L}\left( d_{\textrm{plug-in}}( \widehat{\phi});\phi,p \right) \right] \leq 8B \bar\sigma\, \sqrt{\frac{2 \log(2|\mathcal{X}|)}{I}}. \]
Incidentally, this means that the worst-case regret of the plug-in rule will tend to be small, provided $\mathcal{X}$ is not large relative to $I$. 

A natural question to ask is whether there exists a matching lower bound for the worst-case regret. We think that it might be possible to give a positive answer to this question, at least in the case in which the dimension of $\mathcal{X}$ is fixed. In particular, we think that a lower bound for minimax regret of order $I^{-1/2}$ could be derived using a similar argument to the one used in the proof of Theorem 1 in  \cite{bartlett1998}, which provides a minimax lower bound for the regret of empirically designed vector quantizers.

\section{Illustrative Example} 
\label{section:examples}
In this section we present a simple example to illustrate how the solutions of the archetype discovery problem can be obtained by clustering the values of the function $\phi$.\footnote{Appendix \ref{subsec:dp-illustration-Example2} presents an additional illustrative example using data from the Atlantic Causal Inference Conference (ACIC) 2016.} Our goal is twofold. First, we want to provide the reader a concrete sense of the computational cost associated to the use of the dynamic programming algorithm for clustering the $N$ different values of $\phi$. Second, we would like to suggest different ways of visualizing the archetype set and the values of the function $\bar{\phi}$. A distinction that is useful to keep in mind while reading this section is that there is usually a difference between the cardinality of the \emph{domain} of $\phi$ (denoted $|\mathcal{X}|$) and the cardinality of the \emph{image} of $\phi$ (that is, the number of different values that $\phi$ takes, denoted by $N$). With this distinction in mind, we present tables and figures with information on both \(|\mathcal{X}|\) and \(N\). We assume throughout that $p(x)=1/|\mathcal{X}|$ for every $x \in \mathcal{X}$. Thus, the optimized oracle loss is
\[
L_{oracle}(c) \equiv \sum_{i=1}^{N} p_i \bigl(\phi_i-\mu_{c(i)}\bigr)^2 \qquad \text{where}
\qquad
p_i=\frac{|\{x \in \mathcal{X}\: | \: \phi(x)=\phi_i\}|}{|\mathcal{X}|}.
\]

\subsection{Example 1}
\label{subsec:dp-illustration-Example1}

Consider the function
\[
\phi(x_1,x_2)=\exp\!\left(-(x_1^2+x_2^2)\right),
\qquad
(x_1,x_2)\in[-1,1]^2,
\]
evaluated on an equally spaced \(300\times 300\) grid over \([-1,1]^2\). Hence \(|\mathcal X|=90{,}000\), \(N=7{,}401\), and \(K=10\). 

We find the oracle solution to the archetype discovery problem using \cite{bruce1965optimum}'s dynamic programming algorithm described in Section \ref{subsec:dynamic_prog}. The end-to-end computational complexity of this algorithm is known to be
\[
O\!\bigl(|\mathcal X|\log |\mathcal X| + K N^2\bigr).
\]
The first component comes from sorting the values of $\phi(\mathcal{X})$, which requires computations of order \(O(|\mathcal X|\log |\mathcal X|)\). The second component is the cost of the dynamic-programming routine, which is known to be of order \(O(KN^2)\). The exact dynamic program has runtime of \(0.3688\) seconds, and the value of the oracle solution is \(L_{oracle}=0.000575\).

Panel a) of Figure \ref{fig:dp-sim-top-sorted} below reports the archetype sets (each in different color) obtained by the dynamic programming algorithm. As discussed in Section \ref{section:basic_framework}, the archetype sets inherit the structure of the level sets of the original function \(\phi\) (which, in this example, are concentric rings). Panel a) presents the archetype sets. Panel b) reports the sorted values of the image of $\phi$, where each different color represents a different cluster. The horizontal lines depict the within cluster means, which correspond to the values of the function $\bar{\phi}$ for each $x$ that leads to values of $\phi$ in that cluster. 

\begin{figure}[!h]
    \centering
    \includegraphics[width=0.96\textwidth]{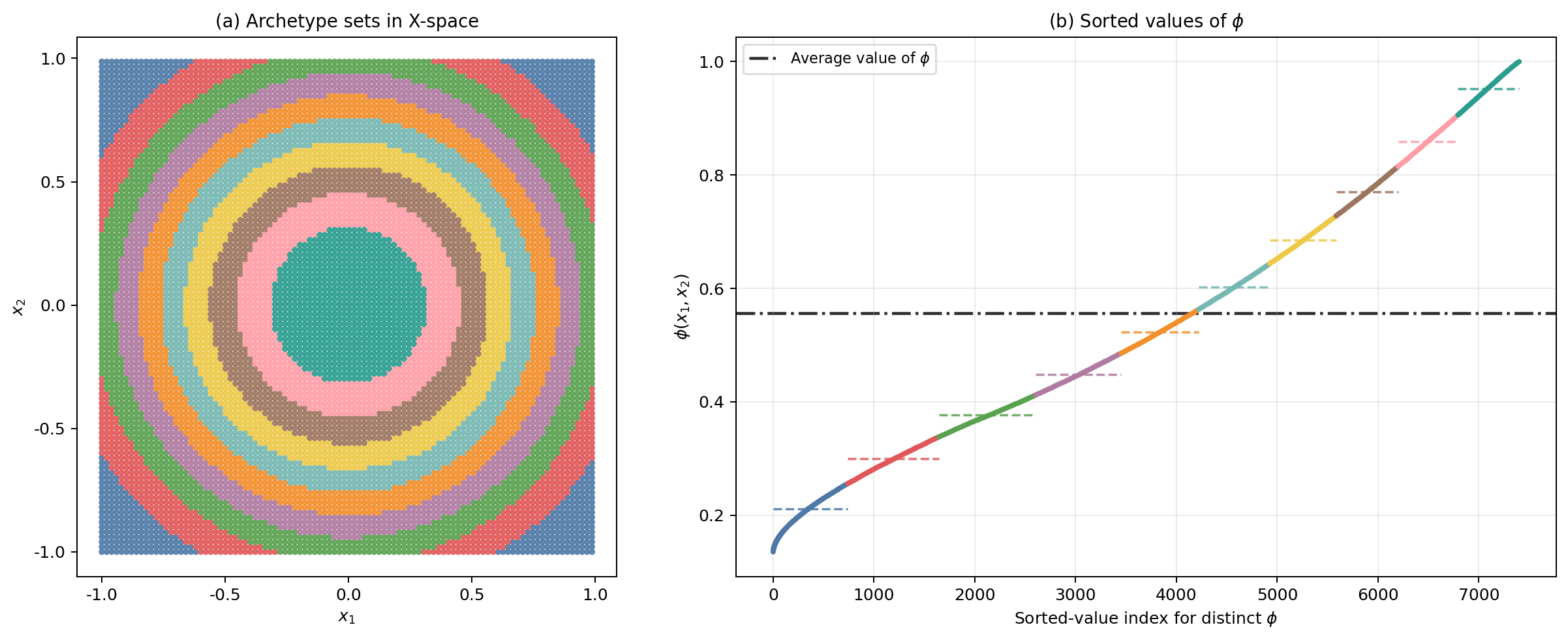}
    \caption{Panel a): Archetype sets of the function $\phi(x_1,x_2)=\exp(-(x_1^2+x_2^2))$. Panel b): sorted and clustered values of $\phi$. Number of archetypes is $K=10$. }
    \label{fig:dp-sim-top-sorted}
\end{figure}

Figure \ref{fig:dp-sim-cluster-mass} complements Figure \ref{fig:dp-sim-top-sorted} by showing the within-cluster variance and the share of covariate values contained in each cluster. The color coding for each of the bars is the same as the one used to depict the different clusters in Figure \ref{fig:dp-sim-top-sorted}. Note that the clusters are indexed according to the value of their within-cluster mean (the first cluster, containing 6.8\% of the observations, has the lowest within-cluster mean).

\begin{figure}[!htbp]
    \centering
    \includegraphics[width=0.72\textwidth]{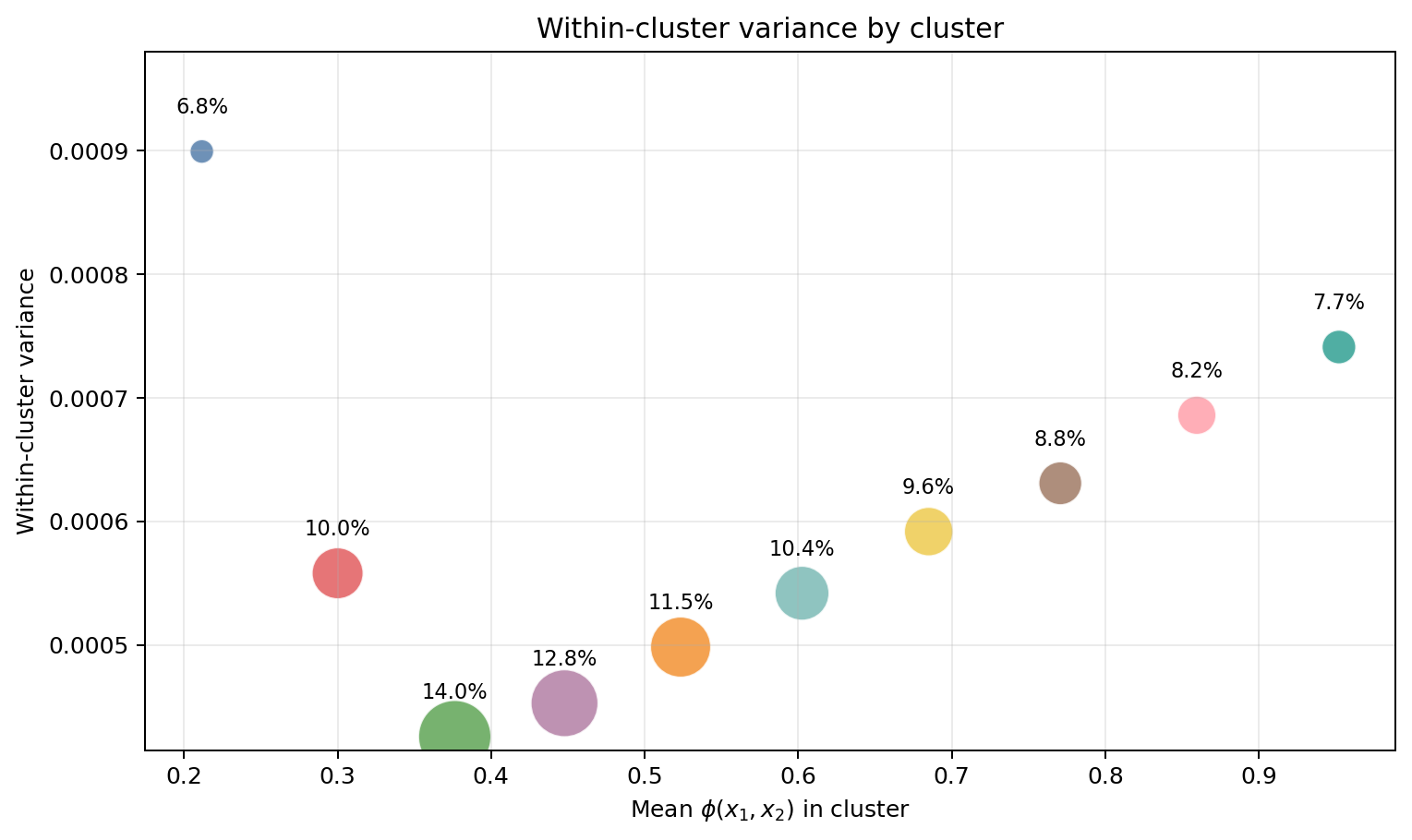}
    \caption{Within-cluster variance for the $K=10$ partition in Example 1. Bubble area is proportional to the cluster's proportion of grid points, and the labels $1$--$10$ denote the displayed cluster indices.} 
\label{fig:dp-sim-cluster-mass}
\end{figure}

\section{Extensions} 
\label{section:extension}
\subsection{Abstention at the cluster level} \label{subsection:abstention_main}
Consider now the case in which---in addition to communicating the function $\bar{\phi}$ to the policymaker---the researcher can also provide a set of covariate values at which, instead of making a prediction, it is better to \emph{admit ignorance} or \emph{abstain}. We follow \cite{breza2025generalizability} and model abstentions as a function $\pi: \mathcal{X} \rightarrow \{0,1\}$, where $\pi(x)=0$ is interpreted as a suggestion (from the researcher to the policymaker) to avoid using the function $\bar{\phi}$ for making predictions at $x$. We adopt the same loss function as in \cite{breza2025generalizability}:     
\begin{equation}\label{eq:loss_abstention}
L(\bar\phi, \pi; \phi,p) \;\equiv\; \sum_{x\in\mathcal{X}} p(x) \Big[ \pi(x) \big(\phi(x) - \bar\phi(x)\big)^2 + (1-\pi(x))\, \sigma^2 \Big],
\end{equation}
where $\sigma^2$ denotes the cost of abstention. 

In order to solve the oracle problem we impose two mild restrictions on the functions $\bar{\phi}$ and $\pi$. First, we assume that the researcher is only allowed to report a summary $\bar{\phi}$ that belongs to the set $\bar{\Phi}^*_K(A_K,G^{\phi})$ defined in Equation  \eqref{eq:measurable_w.r.t_G}.\footnote{That is, the function $\bar{\phi}$ is required to be measurable by the smallest $\sigma$-algebra that makes $\phi$ measurable.} As discussed in Remark \ref{rem:measurability}, this restriction does not entail any loss of generality when abstention is not allowed and we impose it to have a more direct comparison with Theorem \ref{thm:archetype discovery solution}.  

Second, we make a restriction on the type of abstention that the researcher can recommend to the policy maker. Define 
\[ \Pi_K(\bar\phi) \equiv \{ \pi: \mathcal{X} \to \{0,1\} \: | \:   \bar\phi(x)=\bar\phi(x') \implies \pi(x)=\pi(x')\}.\]
Note that an abstention function $\pi \in \Pi_K(\bar\phi)$ \emph{respects}  the archetype sets defined by $\bar{\phi}:$ if the researcher recommends abstention for one covariate value in the $k$-th archetype set, then the researcher must necessarily recommend abstention for all the covariate values in such set. We think this is a reasonable restriction (consistent with the idea that the policymaker has some limited ability to parse complex functions). 

Define the oracle solution to the archetype discovery problem with $\bar{\phi}$-abstention as 
\begin{equation} \label{eq:archetype discovery problem abstention}
\inf_{\bar{\phi} \in \bar{\Phi}^*_K(A_K,G^{\phi}),\pi \in \Pi_K(\bar\phi)} L(\bar\phi, \pi; \phi,p).
\end{equation}

\begin{theorem}\label{thm:thm_oracle_abstention_simple}  
Let $i:\mathcal{X} \rightarrow \{1,\ldots, N\}$ be the function such that $\phi(x) = \phi_{i(x)}$. If $c^*$ solves the clustering problem 
\begin{equation} \label{eq:oracle_abstention} 
\min_{c:\{1,\ldots, N\} \rightarrow \{1,\ldots, K\}} \sum_{k=1}^{K} \min \left\{ \sum_{\{i \: | \: c(i)=k \}}  p_i \Big[   \big(\phi_i - \mu_k(c) \big)^2 - \sigma^2  \Big],0 \right\},    
\end{equation}
then the function 
\[ \bar{\phi}^*(x) = \mu_{c^*(i(x))}(c^*)   \]
solves \eqref{eq:archetype discovery problem abstention}. The oracle abstention function is
\[ \pi^*(x) \equiv \mathbf{1} \left \{ \sum_{\{i \: | \: c(i)=c^*(i(x)) \}}  p_i \big(\phi_i - \mu_{c^*(i(x))}(c^*) \big)^2     \leq \sigma^2 \sum_{\{i \: | \: c(i)=c^*(i(x)) \}}  p_i \right \}. \]
\end{theorem}

\begin{proof}
See Appendix \ref{subsec:proof_oracle_abstention}
\end{proof}

Note that when $\sigma^2$ is large enough---for example, when $\sigma^2 \geq (\phi_N-\phi_1)^2$---a solution to the problem in \eqref{eq:oracle_abstention} can be obtained by solving the usual $K$-means clustering problem in \eqref{eq: k-means clustering}.\footnote{In fact, Popoviciu's inequality \citet{Popoviciu1935,BhatiaDavis2000} implies that a weaker sufficient condition is $\sigma^2 \geq (\phi_N-\phi_1)^2/4$.} However, it is relatively straightforward to provide examples in which solving \eqref{eq: k-means clustering} (and declaring abstention whenever the cluster variance is larger than $\sigma^2$) is not optimal. For instance, suppose that we want to cluster the values $\{0,2,3,5\}$ into $K=2$ clusters, and suppose further that the weights are uniform. Assume also that the cost of abstention is $\sigma^2=1$. Algebra shows that the clusters $\{0,5\}$ and $\{2,3\}$ (with abstention for the elements in the first cluster) evaluate to a lower value of the objective function in \eqref{eq:oracle_abstention} than the $K$-means clustering solution in \eqref{eq: k-means clustering}, which consists of the contiguous clusters $\{0,2\}$ and $\{3,5\}$. To see this, note that the objective function in \eqref{eq:oracle_abstention} for clusters $\{0,5\}$ and $\{2,3\}$ equals 
\begin{eqnarray*} 
&& 0 + \min\{ (1/4)(2-5/2)^2 + (1/4)(3-5/2)^2 - 1/2,0  \} = - 3/8.
\end{eqnarray*}
But the objective function in \eqref{eq:oracle_abstention} for clusters $\{0,2\}$ and $\{3,5\}$ equals 0. This means that declaring abstention for values $\{0,5\}$ and clustering together $\{2,3\}$, is better than using the $K$-means clusters $\{0,2\}$ and $\{3,5\}$ (both of which leave the researcher indifferent between using the reported cluster means or abstaining).

\noindent \emph{Algorithms for archetype discovery with abstention.} It is not clear to us that---without further restrictions---there exists an algorithm for solving the oracle archetype discovery problem with abstention in \eqref{eq:archetype discovery problem abstention} that runs in polynomial time in $(K,N)$. When the clusters are not required to be contiguous, the problem seems to be combinatorial (as one searches over all possible subsets of $\{\phi_1, \ldots, \phi_{N}\}$ that are good candidates for declaring abstention). In contrast, we can show that if we further require the clusters to be contiguous---namely, if we focus on clustering functions such that $\phi_n < \phi_{n'}< \phi_{n{''}}$ and $c(n)=c(n'')=k$ imply $c(n')=k$---then a minor modification of the flow payoff in the dynamic programming algorithm presented in Section \ref{subsec:dynamic_prog} produces a solution to the   problem with contiguous clusters.

Figure \ref{fig:abstention_topview} presents the results of the archetype discovery problem with abstention in the context of Example 1. Once again, we take the number of clusters to be $K=10$ and require these clusters to be contiguous. We set the abstention cost at $\sigma^2=.0010$. Figure \ref{fig:abstention_topview} shows that under this parameterization, the researcher admits ignorance for the covariates associated to the \emph{largest} values of the policy effects. This result is not ex ante obvious: according to Figure \ref{fig:dp-sim-cluster-mass}, the cluster with the \emph{smallest} policy effects has the largest within-cluster variance. One simple intuition for our findings has to do with the fact that, despite the large within-cluster variance, the left-most cluster in Figure \ref{fig:dp-sim-cluster-mass} has the smallest share of observations and therefore does not have a big effect on the loss function. 

\begin{figure}[!htbp]
    \centering
    \includegraphics[width=1\textwidth]{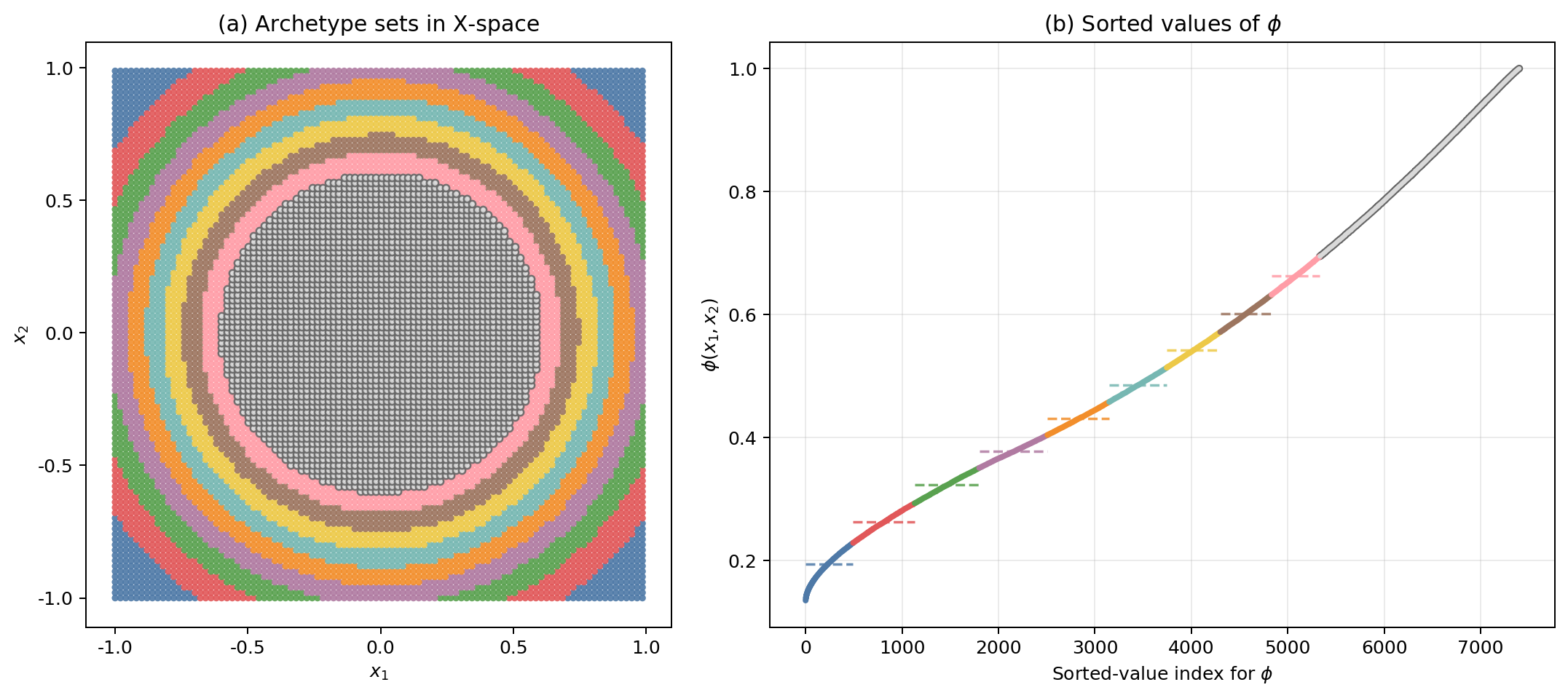}
    \caption{Top view of the oracle solution with abstention at $K=10$ and $\sigma^2=0.0010$. Panel (a) shows the archetype sets in $X$-space. The light-gray central region is the abstained cluster. Panel (b) shows the same solution on the sorted values of $\phi$ while the abstained group is displayed in gray.}
\label{fig:abstention_topview}
\end{figure}

Figure \ref{fig:kmeans_with_without_abstention} confirms this intuition and presents a more detailed comparison of the oracle solutions with and without abstention. The share of observations that belong to the right-most cluster (for which the researcher admits ignorance) is about 30\%. This cluster seems to contain the three right-most clusters that arise when abstention is not allowed. In the problem without abstention, the left-most cluster exhibited a high intra-cluster variance. This variance is considerably reduced in the problem with abstention.

\begin{figure}[!htbp]
    \centering
    \includegraphics[width=1\textwidth]{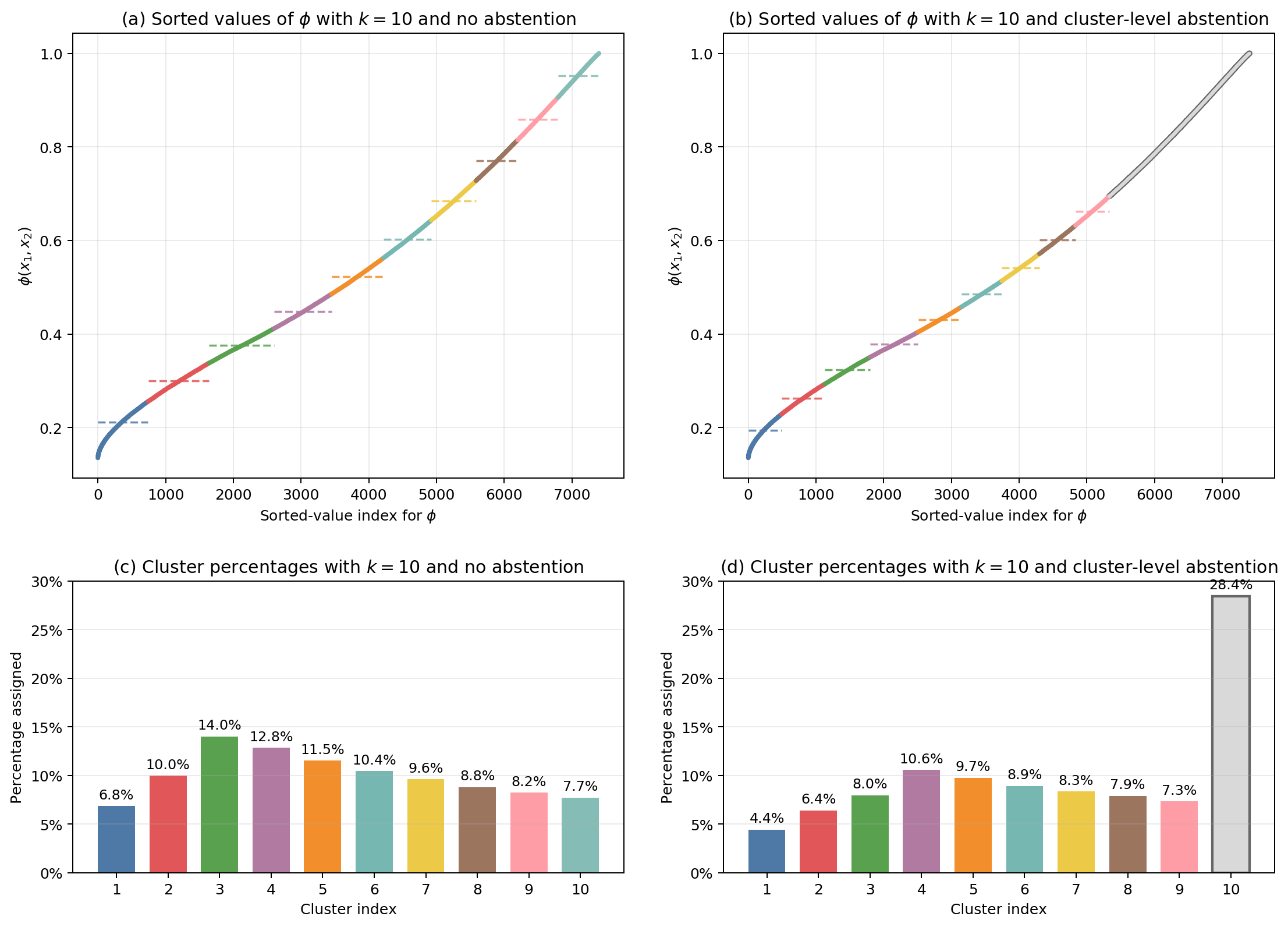}
    \caption{Comparison between the oracle solution ($K=10$) with and without abstention. The abstaining cost is chosen at $\sigma^2=0.0010$. Panel (a) and (b) compares the archetype discovery problem solution on top of sorted values of $\phi$ with and without abstention. Panel (c) and (d) show the corresponding percentages of the assigned cluster. In panel (b) and (d), the light-gray curve segment and the light-gray bar represent the abstained cluster.}
\label{fig:kmeans_with_without_abstention}
\end{figure}

Finally, Figure \ref{fig:abstention_sigma_grid} reports the effects of varying $\sigma^2$ over the archetype sets. As the parameter $\sigma^2$ increases (and admitting ignorance becomes more costly), the researcher admits ignorance for fewer and fewer values of $\phi$. Note that at $\sigma^2=.0002$ (the smallest value of $\sigma^2$ we consider in the upper graph of the figure), the researcher admits ignorance for both the smallest and the largest values of $\phi$. As we discussed before, these are the groups of observations with the largest within-cluster variance. 

\begin{figure}[!htbp]
    \centering
    \includegraphics[width=1\textwidth]{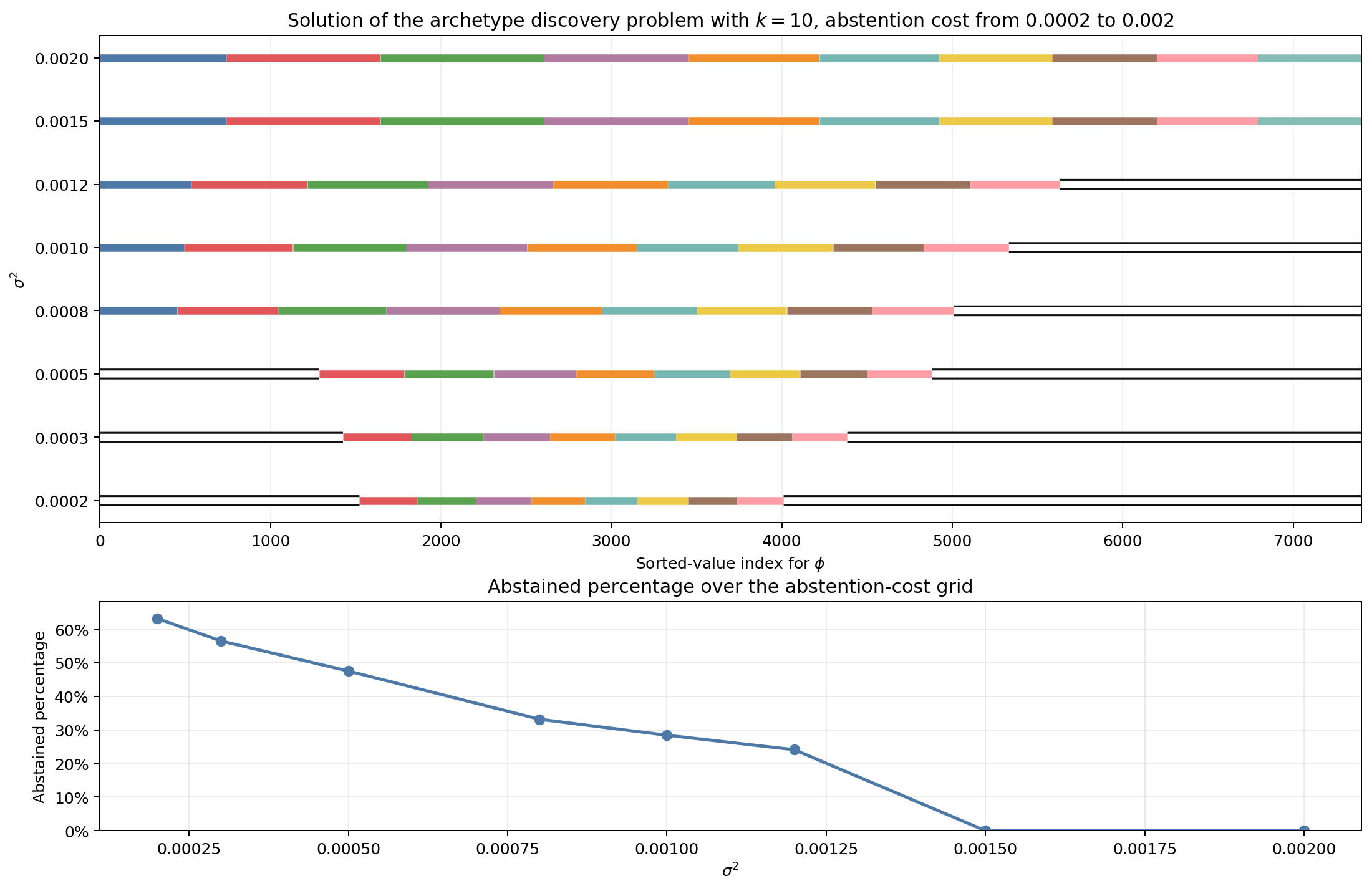}
    \caption{Solution of the archetype discovery problem with $K=10$ over the abstention-cost grid $\sigma^2 \in \{0.0002, 0.0003, 0.0005, 0.0008, 0.0010, 0.0012, 0.0015, 0.0020\}$. In the top panel, colored segments show the cluster assignment, while white segments with black outlines denote abstained clusters. The abstained sets are $\{1,10\}$ for $\sigma^2=0.0002,0.0003,0.0005$, $\{10\}$ for $\sigma^2=0.0008,0.0010,0.0012$, and none for $\sigma^2=0.0015,0.0020$. The bottom panel reports the percentage of points assigned to abstention at each value of $\sigma^2$.}
\label{fig:abstention_sigma_grid}
\end{figure}

\subsection{Other algorithms}
\label{subsec:benchmark-comparison}

We now return to Example 1 in Section ~\ref{subsec:dp-illustration-Example1} and compare the exact dynamic-programming solution with two popular algorithms. The first alternative algorithm to summarize $\phi$ is a (CART) tree as described in  \cite{breiman2017classification}. The tree partitions the space of covariates by greedy recursive axis-aligned splits and is restricted to \(K\) terminal leaves. Recursive decision trees implemented using CART-type recursive partitioning are now widely used to estimate heterogeneous causal treatment effects in experimental and observational studies (although it has been shown recently by \cite{cattaneo2025honest} that causal trees constructed via standard CART-type partitioning may exhibit poor convergence properties). In our example, we are assuming $\phi$ is known, and we simply want to compare the archetype sets generated by the trees to those generated by the oracle solution.  In our implementation, the computational complexity of implementing the tree is
\(
O\!\left(Kd|\mathcal X|\log|\mathcal X|\right),
\)
where \(d\) is the dimension of the covariate vector.

The second algorithm is Lloyd's algorithm \cite{lloyd1982least}. 
It targets the scalar clustering problem more directly, but it is a local-search heuristic and does not generally guarantee the global minimizer of~\eqref{eq: k-means clustering}.
Its computational complexity is
\(
O\!\left(TK|\mathcal X|\right),
\)
where \(T\) is the number of iterations until convergence. 

For the CART benchmark, each terminal leaf is assigned the \(p\)-weighted average of \(\phi\) over that leaf, and the resulting partition is evaluated using the same oracle loss \(L_{oracle}\) used in Section~\ref{section:examples}. 
Figure~\ref{fig:tree-sim-top-sorted} displays the resulting partition. 
The left panel shows the induced partition in covariate space, while the right panel maps the same assignment onto the sorted values of \(\phi\).

\begin{figure}[!htbp]
    \centering
    \includegraphics[width=0.96\textwidth]{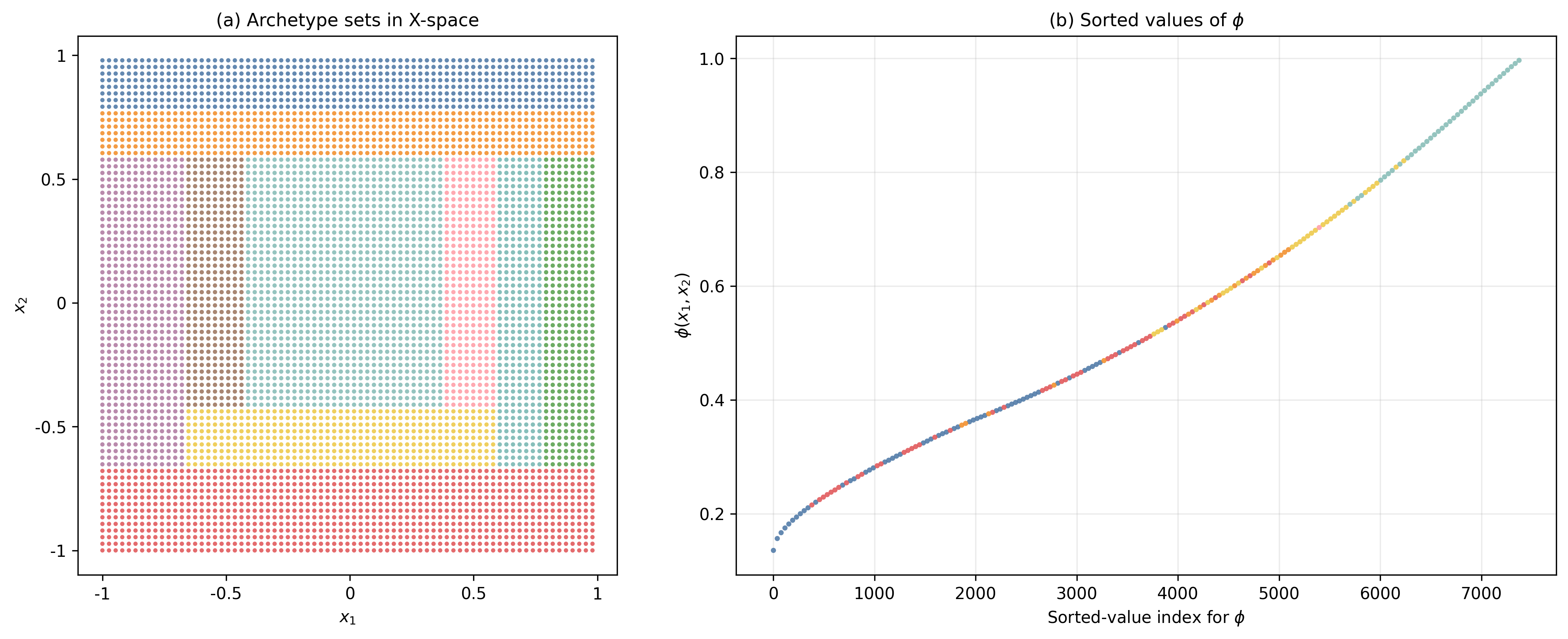}
    \caption{CART tree benchmark for Example~\ref{subsec:dp-illustration-Example1}. 
    The left panel displays the induced partition in covariate space. 
    The right panel displays the induced assignment on the sorted values of \(\phi\). 
    The tree has \(K=10\) terminal leaves.}
    \label{fig:tree-sim-top-sorted}
\end{figure}

The figure shows the geometric restriction imposed by the tree. 
The exact archetype sets in Example 1 in Section~\ref{subsec:dp-illustration-Example1} inherit the circular geometry of the level sets of \(\phi\). 
By contrast, the tree approximates these sets by recursive axis-aligned rectangles. 
When the same assignment is viewed on the sorted values of \(\phi\), it is fragmented rather than contiguous. 

We next consider Lloyd's algorithm. 
Unlike the tree, Lloyd's algorithm is applied directly to the scalar values of \(\phi\). 
It is therefore closer to the clustering problem characterized by Theorem~\ref{thm:archetype discovery solution}. 
Figure~\ref{fig:lloyd-sim-top-sorted} shows the resulting partition in covariate space and on the sorted values of \(\phi\).

\begin{figure}[!htbp]
    \centering
    \includegraphics[width=0.96\textwidth]{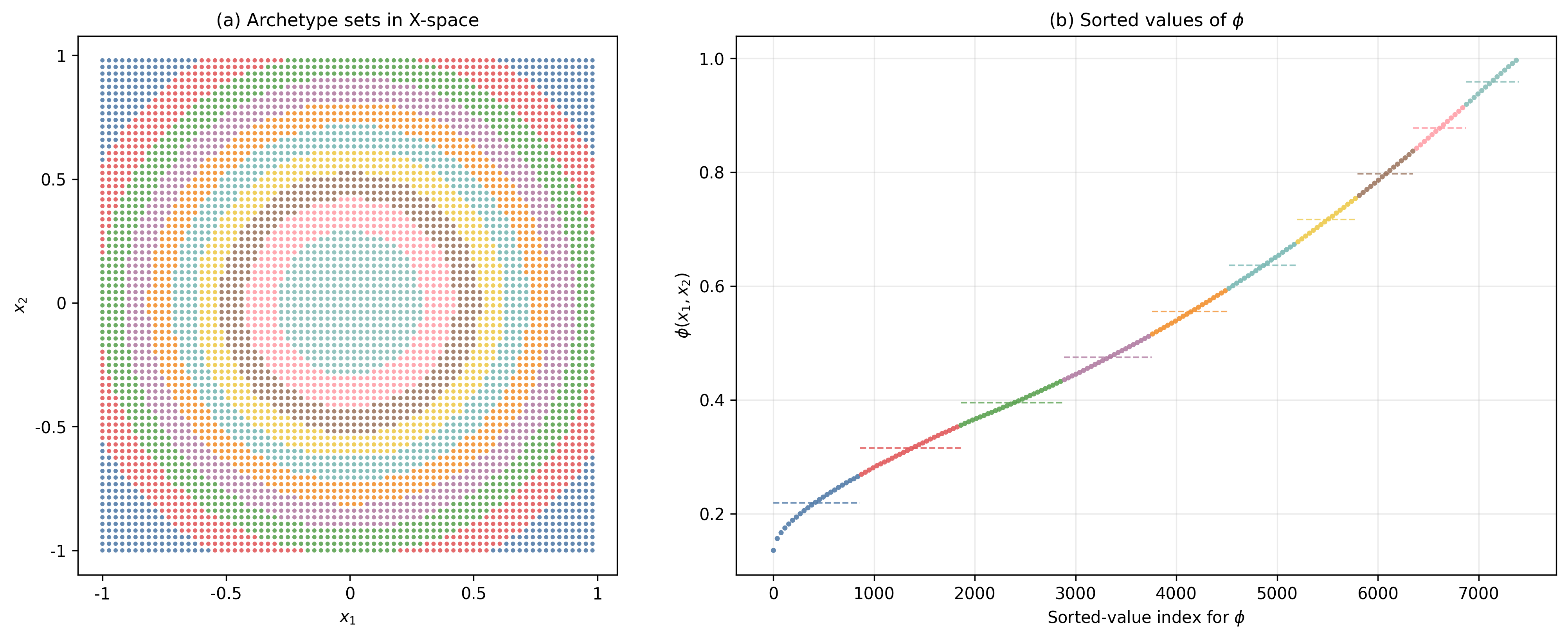}
    \caption{Lloyd's algorithm for Example~\ref{subsec:dp-illustration-Example1}. 
    The left panel displays the induced partition in covariate space. 
    The right panel displays the induced assignment on the sorted values of \(\phi\). 
    The algorithm is implemented with \(K=10\), a single random initialization, and a single run.}
    \label{fig:lloyd-sim-top-sorted}
\end{figure}

The clustering based on Lloyd's algorithm is much closer to the interval structure delivered by the exact dynamic program. 
This is expected: both procedures operate on the scalar values of \(\phi\), rather than imposing a recursive partition on \(x\)-space. 
The numerical closeness, however, should not be interpreted as an optimality guarantee. 
Lloyd's algorithm returns the fixed point reached from its initialization and update path, whereas the dynamic program computes the global minimizer of~\eqref{eq: k-means clustering}.

Table~\ref{tab:benchmark-algorithm-comparison} summarizes the comparison. 
The table reports two exact algorithms and two benchmark procedures. 
The first exact algorithm is Bruce's dynamic program, which is the baseline implementation used in Example~\ref{subsec:dp-illustration-Example1}. 
The second is \texttt{Ckmeans.1d.dp}, a faster exact implementation of one-dimensional \(K\)-means \cite{wang2011ckmeans}; for the weighted case considered here, we use the corresponding weighted extension \cite{song2020efficient}. 
Both exact algorithms solve the same one-dimensional clustering problem and therefore attain the same oracle loss. 
Their difference is computational rather than statistical or decision-theoretic.

\begin{table}[!htbp]
\centering
\caption{Algorithm comparison for Example~\ref{subsec:dp-illustration-Example1}.}
\label{tab:benchmark-algorithm-comparison}
\setlength{\tabcolsep}{6pt}
\renewcommand{\arraystretch}{1.15}
\begin{tabular}{lp{0.40\textwidth}cc}
\hline
Algorithm & Complexity & oracle loss & Runtimes \\
\hline
Exact DP (Bruce DP)
& \(O\!\left(|\mathcal X|\log|\mathcal X|+KN^2\right)\) 
& 0.000575 
& 0.3688 \\

Exact DP (\texttt{Ckmeans.1d.dp}) 
& \(O\!\left(|\mathcal X|\log|\mathcal X|+KN\right)\) 
& 0.000575 
& 0.0063 \\

CART tree 
& \(O\!\left(Kd|\mathcal X|\log|\mathcal X|\right)\), \(d=2\) 
& 0.009421 
& 0.0315 \\

Lloyd's algorithm 
& \(O\!\left(TK|\mathcal X|\right)\), \(T=49\) 
& 0.000593 
& 0.0743 \\
\hline
\end{tabular}

\medskip
\begin{minipage}{0.94\textwidth}
\footnotesize
\emph{Notes.}
All entries use the simulation design in Example~\ref{subsec:dp-illustration-Example1}, where \(|\mathcal X|=90{,}000\), \(N=7{,}401\), and \(K=10\). 
For each procedure, the reported value within a cluster is the \(p\)-weighted average of \(\phi\) over that cluster, and the displayed objective value is the corresponding value of \(L_{oracle}(c)\). 
Bruce DP refers to the optimum-quantization dynamic program of \cite{bruce1965optimum}. 
\texttt{Ckmeans.1d.dp} refers to the exact one-dimensional implementation of \cite{wang2011ckmeans}; the weighted extension is due to \cite{song2020efficient}. 
The CART benchmark is a greedy regression tree with \(K\) terminal leaves \cite{breiman2017classification}. 
Lloyd's algorithm is implemented with a single random initialization and a single run \cite{lloyd1982least}. 
\end{minipage}
\end{table}

The comparison highlights four points. 
First, exact one-dimensional clustering is computationally feasible at the scale of the examples considered in this paper. 
Second, faster exact implementations such as \texttt{Ckmeans.1d.dp} preserve the same oracle loss while reducing runtime substantially. 
Third, the CART tree has a much larger oracle loss in this example. 
Fourth, Lloyd's algorithm is much closer to the exact objective, but its performance remains that of a heuristic benchmark rather than an exact characterization of the oracle archetype report.

\section{Application} \label{sec:application}
In order to illustrate the usefulness of our theoretical results, we revisit the application discussed in \cite{Victor25}: an experiment with the government of Haryana in North India designed to analyze the effects of a policy bundle that provided different incentives for immunization across several villages. The goal of the intervention was to increase the takeup of immunization services. In particular, the outcome of interest discussed in \cite{chernozhukov2025reply} is the number of children (15 months or younger in a given month in a given village) that completed all the vaccines in the immunization schedule.  

In their data, there are 25 villages that were randomly assigned to the policy bundle, and 78 control villages. Almost all the treatment and control villages are followed for 12 months (the duration of the intervention).\footnote{See p. 1150 and 1151 of \cite{Victor25} for specific details on the treatment.} The total number of village-months observations is 843. Their replication package has 42 covariates available, 39 of which are baseline-village characteristics (such as religion, caste, financial status, baseline immunization, etc). The remaining variables include fixed effects, cluster indicators, and external weights for each village.\footnote{Based on our reading of \cite{chernozhukov2025reply} and their replication package: fixed effect index district-by-year-month cells, capturing the local time and district context of each observation; the cluster indicator identifies the village, since villages appear repeatedly across months; finally, external weight measures village size and is used as the population weight in the original weighted analysis.}   

The first two graphs below report the sorted values of the estimators---or \emph{machine learning (ML) proxies}---of the conditional average treatment effects using both an Elastic Net and a Neural Network. The sorted values reported in Figure \ref{fig:sorted_ML_proxies} are generated based on a random $67\%/33\%$-split of the 843 observations. The smaller sample in the split---henceforth, the \emph{main data}---has 281 observations. The ML proxies, $\widehat{\phi}$, are estimated on the dataset containing $843-281=562$ observations, with 78 villages in the control group and 25 villages in the treatment group. The function $\widehat{\phi}$ is then evaluated at the covariates that appear in the main data. This gives a total of $N=96$ different policy effects, ranging from $-13.803$ to $43.60$ for the Elastic Net estimator, and from $-13.495$ to $23.032$ for the Neural Network estimator. The solid horizontal red-line in Figure \ref{fig:sorted_ML_proxies} represents the median effect of the policy bundle: $2.814$ for the Elastic Net, and 2.441 for the Neural Net. These two numbers are taken from Table III in \cite{Victor25}, where the reported median is taken over 250 random data splits. The $K=5$ groups that we generate using the dynamic programming algorithm in Section \ref{subsec:dynamic_prog} are reported in different colors along with their corresponding summary values. For comparison, we also added black dashed lines to represent the quintiles of the policy effects. Computing the ML proxies and running the weighted $K$-means clustering algorithm takes about 1.965 seconds in a personal laptop. Computing the groups based on quantiles takes about 1.793 seconds. This means that the extra computational burden associated to constructing archetype sets using the weighted $K$-means algorithm is minimal.     

\begin{figure}[!h] 
    \centering
    \includegraphics[width=0.96\textwidth]{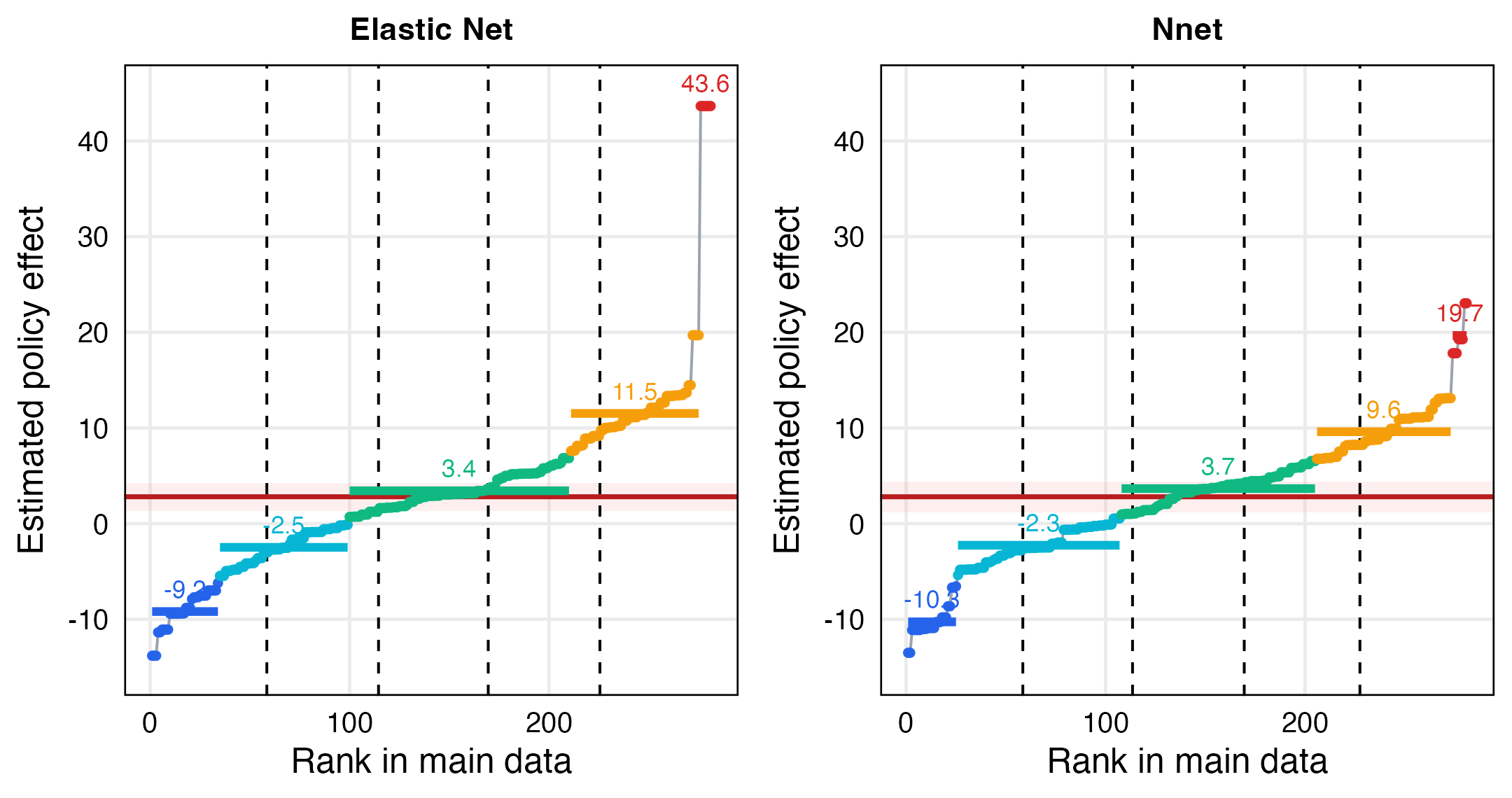} 
    \caption{Graphical representation of the weighted $K$-means clusters $(K=5)$ based on the sorted Elastic Net and Neural Net proxies. The $x$-axis is the index of the sorted ML proxy in the main data (281 total observations). The $y$-axis is the value of the ML proxy. Different colors are used to represent the different clusters. The colored horizontal lines are the summary values of each cluster. The vertical dashed lines represent the $K=5$ groups based on equally-spaced quantiles. The red horizontal line and the shaded red area are the values reported in Table III of \cite{Victor25}. } \label{fig:sorted_ML_proxies}
\end{figure}


Figure \ref{fig:sorted_ML_proxies} shows that the archetype sets generated by weighted $K$-means clustering are markedly different to the groups generated by equally-spaced quantiles of the ML proxies on the specific $67\%/33\%$ data split we consider.\footnote{The weights we use for each value of the proxy share main-sample observations that map to that value. If two village-level observations have the same estimated ML proxy but different village populations, they contribute equally to the weights.
} A natural question to ask is whether the difference in the size of the archetype sets will also be present in other random data splits. In order to answer this question, we consider 250 random $67\%/33\%$ data splits and, for each of them, we compute the share of observations that belong to each archetype set (the share is always relative to the smaller dataset that contains about $33\%$ of the observations). The orange dots in Figure \ref{fig:size_archetype_sets} report the median share of observations contained in each of the $K=5$ groups. 

\begin{figure}[!htbp]
    \centering
    \includegraphics[width=0.96\textwidth]{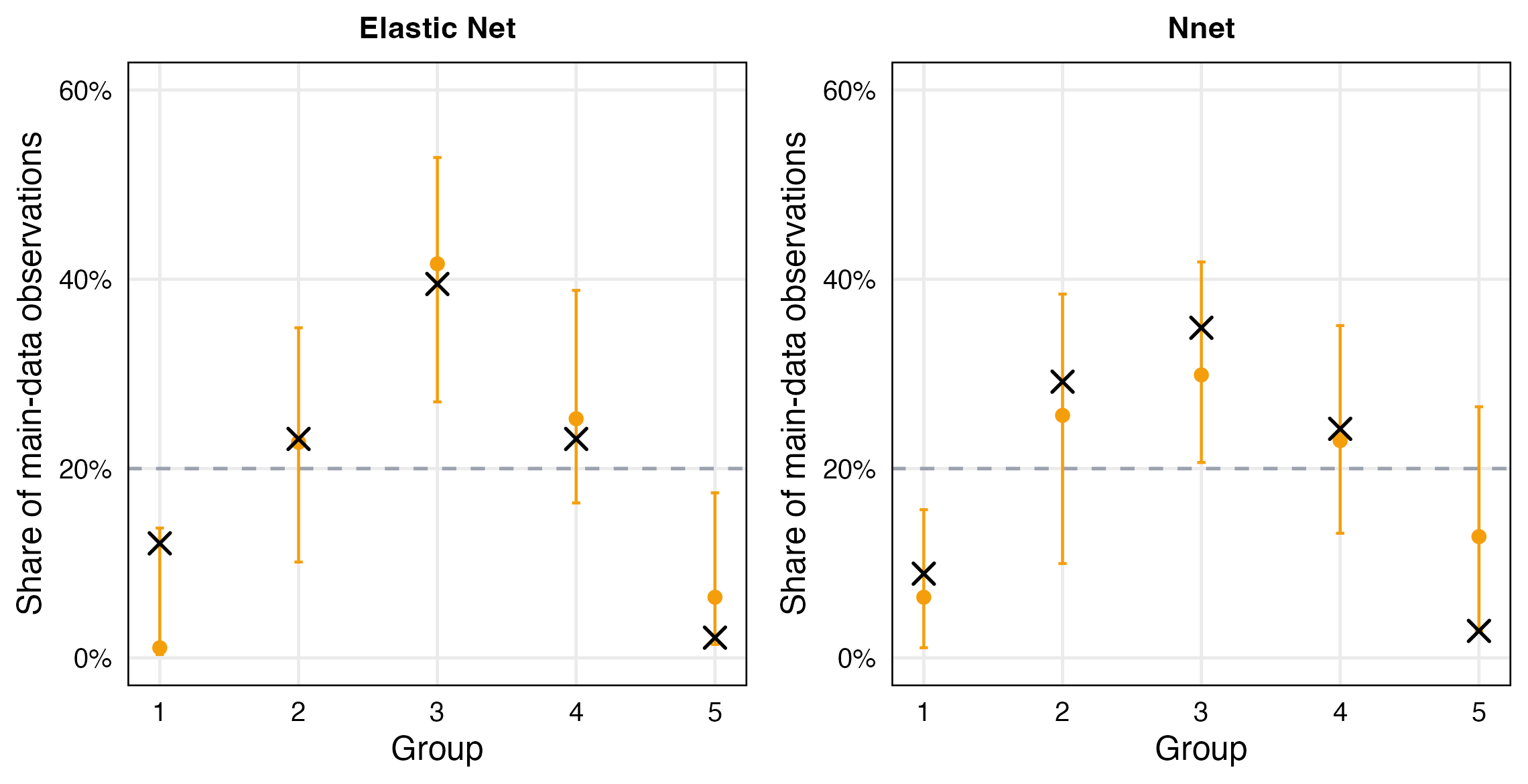}
    \caption{Graphical representation of the size of the weighted $K$-means clusters $(K=5)$ based on the sorted Elastic Net and Neural Net proxies. The $x$-axis is the group index. The $y$-axis is the share of observations in each group (relative to the main data). The groups with the larger indexes have larger values of the ML proxies. Vertical orange lines represent the range of shares observed in 95\% of the 250 random data splits. Black crosses represent the shares corresponding to data split used in Figure \ref{fig:sorted_ML_proxies}. The dashed horizontal line is the size of the groups obtained using quintiles. } \label{fig:size_archetype_sets}
\end{figure}

Figure \ref{fig:size_archetype_sets} confirms the pattern observed in Figure \ref{fig:sorted_ML_proxies}. The group in the center (Group 3) tends to be considerably larger than the size of the group based on the equally-spaced quantiles (which, by construction, contains $20\%$ of the observations). The groups with the smallest and largest values of the proxies (Groups 1 and 5) tend to be smaller when the groups are constructed using weighted $K$-means. The vertical orange lines in Figure \ref{fig:size_archetype_sets} represent the range of shares observed in 95\% of the 250 random data splits (where the data splits corresponding to the lowest 2.5\% values and the highest 97.5\% values have been excluded from consideration). The black crosses in the figure represent the share of observations corresponding to the data split used to construct Figure \ref{fig:sorted_ML_proxies}. 

Since the main object of interest in the archetype discovery problem is the summary of the heterogeneous effects for each archetype set, Figure \ref{fig:treatment_effects} reports the median value of $\bar{\phi}$---the summary of the ML proxies for each group. Once again, the median we report is based on the 250 random data splits that were used in Figure \ref{fig:size_archetype_sets}.   

\begin{figure}[!h]
    \centering
    \includegraphics[width=0.96\textwidth]{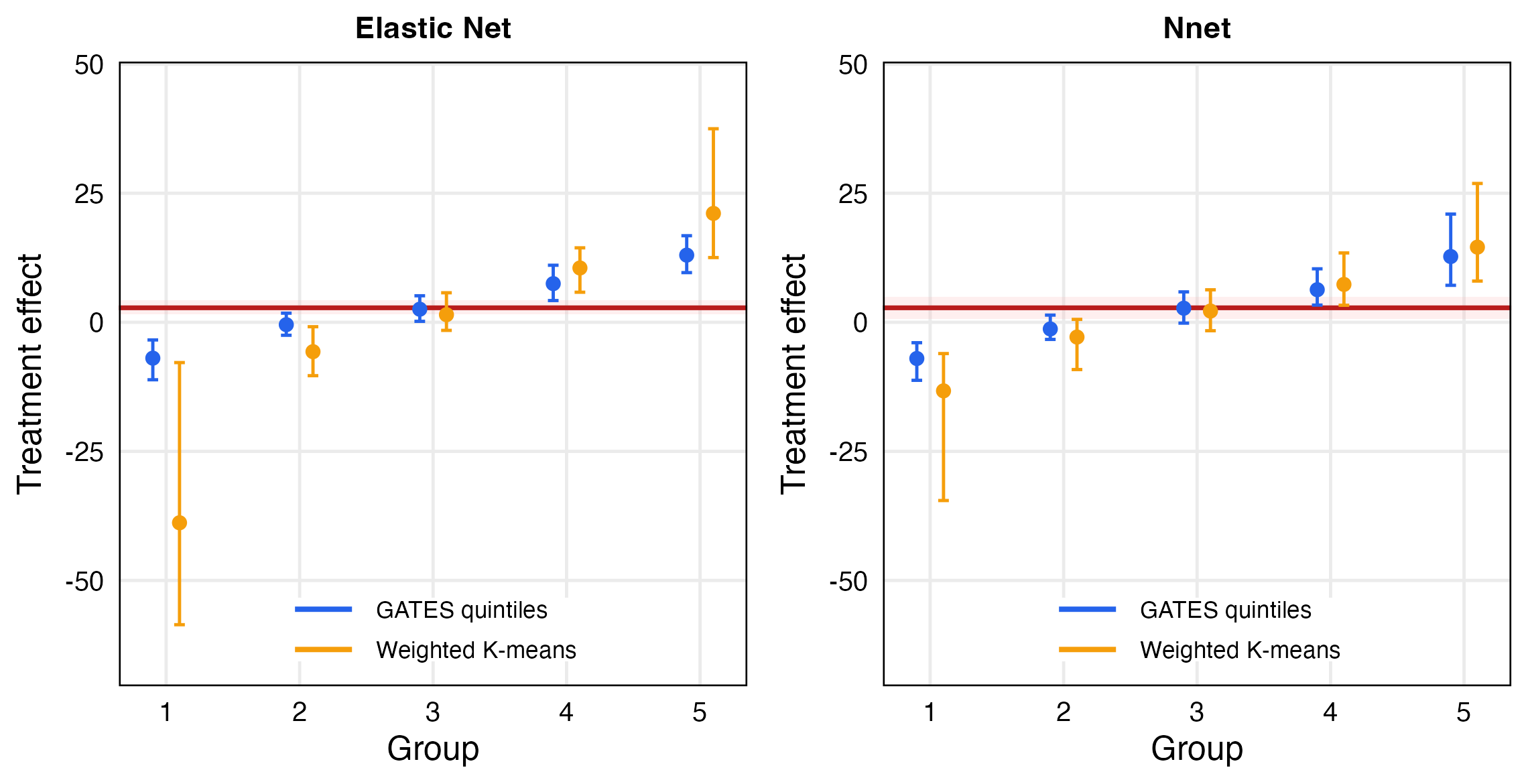}
    \caption{Graphical representation of the average value of the ML proxies for each group. The $x$-axis is the group index. The $y$-axis is the average value of the ML proxies, which we interpret as within-group estimated average treatment effects. The groups with the larger indexes have larger values of the ML proxies. The red horizontal line and the shaded red area are the values reported in Table III of \cite{Victor25}.} \label{fig:treatment_effects}
\end{figure}

The orange dots in Figure \ref{fig:treatment_effects} represent the median value of the researcher's report when the groups are constructed using weighted $K$-means clustering. The blue dots represent the median value of the report when the groups are constructed based on the quintiles of the ML proxy as recommended by \cite{Victor25}. The vertical lines represent the range of values of $\bar{\phi}$ observed in 95\% of the 250 random data splits (where the lowest 2.5\% values and the highest 97.5\% have been excluded from consideration). 

We note that the values of $\bar{\phi}$ based on weighted $K$-means tend to be more extreme---and noisier---for the groups with the smallest and largest values of the ML proxies. Thus, intuition suggests that if we allow the researcher to admit ignorance---as discussed in Section \ref{subsection:abstention_main}---the researcher will likely do so for the groups associated to the smallest and largest values of the ML proxies. Figure \ref{fig:abstention_application} confirms this intuition. In particular, the figure reports the archetype sets for different values of the abstention cost ($\sigma^2$) using the same $67\%/33\%$-data split that was used to construct Figure \ref{fig:sorted_ML_proxies}. The largest value of $\sigma^2$ in the $y$-axis corresponds to the smallest value of $\sigma^2$ for which the researcher \emph{never} admits ignorance. As discussed in Section \ref{subsection:abstention_main}, if we focus on the case where the abstention function \emph{respects the archetype sets}, then we can use a dynamic programming algorithm to solve the oracle archetype discovery problem with abstention.   Figure \ref{fig:abstention_application} shows that as we decrease the abstention cost, the researcher starts to admit ignorance. More importantly, the first groups affected are those corresponding to the smallest and largest values of the proxies. This pattern arises with both ML proxies. 

\begin{figure}[!h]
    \centering
    \includegraphics[width=0.9\textwidth]{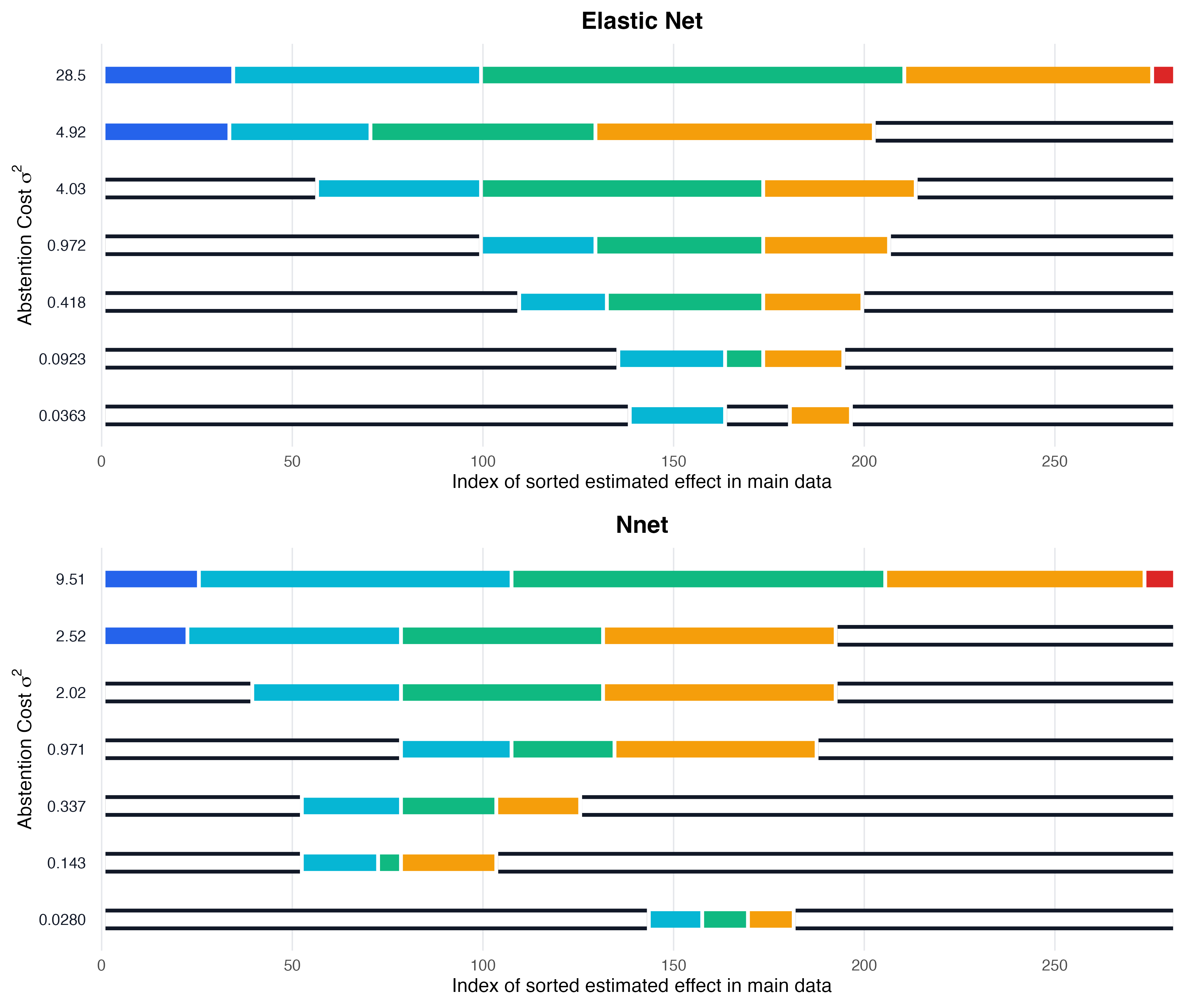}
    \caption{Abstention paths based on the sorted Elastic Net and Neural Net proxies. The $x$-axis is the index of the sorted ML proxy in the main data (281 total observations). Each horizontal row represents the $K=5$ groups associated with the abstention cost $\sigma^2$. Colored segments denote groups for which the researcher provides a summary of the ML proxies, while the white segments denote the groups for which the researcher declares ignorance.}
    \label{fig:abstention_application}
\end{figure}

Finally, we provide a simple illustration of how the archetype sets constructed via weighted $K$-means clustering can be used to complement the Classification Analysis (CLAN) of \cite{Victor25}. As noted by \cite{Victor25}, when the GATES analysis reveals substantial heterogeneity, \emph{``it is interesting to know the properties of the subpopulations that are the most and least affected''}. Using their CLAN methodology, \cite{Victor25} find statistical evidence suggesting that \emph{``the villages with low levels of pretreatment immunization are the most affected by the incentives.''} We focus on two of the baseline covariates that measure pretreatment immunization levels: the fraction of children receiving measles vaccines by 15 months of age, and the fraction of children receiving measles vaccines at credible locations. Figure \ref{fig:baseline_covariates} reports the values of these covariates for each group across 250 random data splits. For each split, we label the groups according to the value of their corresponding summary $\bar{\phi}$: the first group is always the one with the lowest values of the ML proxies (and hence the lowest summary), and the fifth group has always the highest values. For each data split, we calculate the median value of each covariate in each of the groups. The vertical and horizontal lines in each of the panels in Figure \ref{fig:baseline_covariates} represent the range of these medians in 95\% of the 250 random data splits, where the lowest 2.5\% values and the highest 97.5\% have been excluded from consideration. The labeled boxes in each panel represent the median value of the medians. We note that the groups constructed by weighted $K$-means clustering differ very clearly in baseline immunization levels. When the groups are constructed using quintiles of the ML proxy, the top groups seem to be very similar in terms of baseline immunization levels (as measured by the two covariates we consider). Thus, the figure further makes the point that the groups constructed using weighted $K$-means clustering will be very different from the ones constructed using equally-spaced quantiles.       

\begin{figure}[!h]
    \centering
    \includegraphics[width=0.9\textwidth]{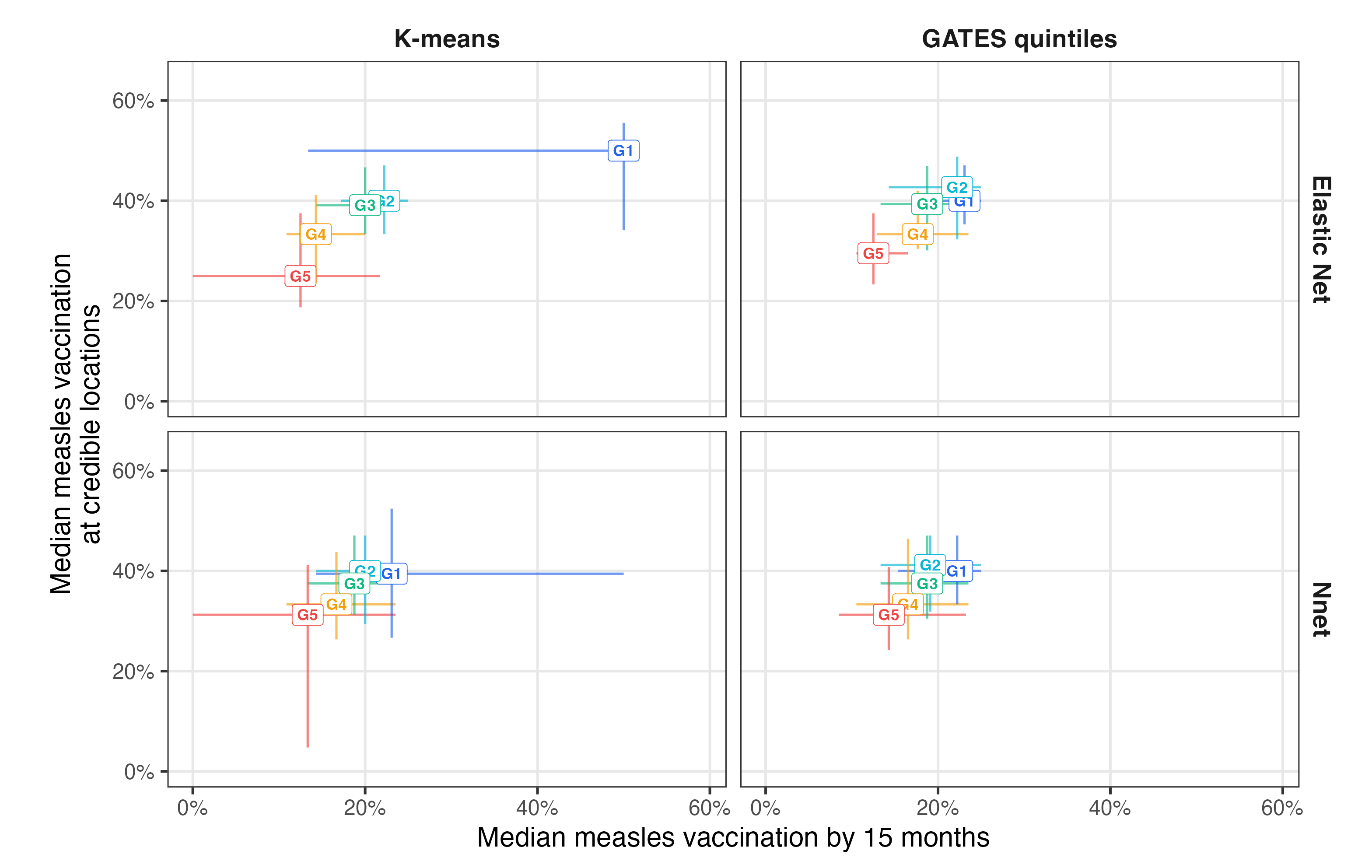}
    \caption{Pretreatment immunization levels of the $K=5$ groups. The $x$-axis is the fraction of children receiving measles vaccines by 15 months of age. The $y$-axis is the fraction of children receiving measles vaccines at credible locations. The groups with the larger indexes have larger values of the ML proxies. The vertical and horizontal lines in each of the panels represent the range of within-group medians of each covariate in 95\% of the 250 random data splits, where the lowest 2.5\% values and the highest 97.5\% of within group medians have been excluded from consideration. The labeled boxes in each panel represent the median value of the medians. }
    \label{fig:baseline_covariates}
\end{figure}

\section{Conclusion} 
\label{section:conclusion}

We used decision theory to analyze the \emph{archetype discovery problem} of \cite{breza2025generalizability}. In this problem, a researcher wants to summarize $N$ heterogeneous policy effects of interest that vary over a discrete set of covariates. The goal is to partition the set of covariates into $K<N$ groups---the \emph{archetype sets}---and to provide a summary of the policy effects for each group. Since there are different recommendations in the literature regarding how to present and summarize heterogeneous policy effects, decision theory is useful to guide their evaluation. 

The main message of this paper is that, under a  weighted mean-squared-error criterion, a procedure analogous to the \emph{Sorted Group Average Treatment Effects} (GATES) of \cite{Victor25} \emph{solves} the archetype discovery problem (in a sense we made precise through our main theoretical results). The key difference is that, in the optimal procedure, the archetype sets are obtained by weighted $K$-means clustering of the $N$ heterogeneous policy effects, instead of relying on $K$ equally-spaced quantiles. An important observation here is that clustering is based on the $N$ scalar policy effects, not on the covariates, as in the recent work of \cite{kim2026causal}. 

A key component of statistical decision theory is the analysis of the risk function of a statistical decision rule, which requires the specification of an action space, a loss function, and a statistical model. As mentioned in \cite{chamberlain2007decision}, \emph{``once one focuses on a risk function, it is natural to think about criteria like average risk and maximum risk that lend themselves to optimization.''} In this paper we analyzed both criteria. In terms of average risk, we showed that the procedure that minimizes average risk for a given prior can be obtained by clustering the different values of the posterior mean estimate of the policy effects of interest. This is a very general result, compatible with nonparametric statistical models and general priors. In terms of maximum or worst-case risk, we introduced a statistical model where the researcher observes an estimator of the policy effects of interest. We showed that an approximately minimax procedure in large samples can be obtained by clustering the estimator of the policy effects. For both the average and maximum risk, the $K$ archetype sets can be found using a simple, and well-known, dynamic programming algorithm. We also analyzed the worst-case regret of the approximately minimax procedure and provided an upper bound on the rate at which it converges to zero. 

While there are other recent papers in the literature that have suggested the use of $k$-means clustering for summarizing policy effects---see the recent work of \cite{kim2026causal}---we are not aware of any other paper in the literature showing that clustering the scalar policy effects (and not the covariates) has some decision-theoretic optimality, despite the high number of citations associated to the GATES procedure and the work of \cite{Victor25}.\footnote{The 2018 working paper version of \cite{chernozhukov2018generic} has more than 800 citations in google scholar.} 

We think there are several areas for future work. First, throughout this paper, we took the desired number of archetypes as given. It would be interesting to consider the possibility of using data-driven strategies for choosing the number of archetypes. See also Remark 3.1 in \cite{kim2026causal}. Second,  the recent work of \cite{nath2026inference} has developed an inferential framework for constructing confidence sets for cluster labels. It would be interesting to use such a framework to think about how to conduct inference about archetype sets or to adopt the randomization inference framework of \cite{imai2025comment,imai2025statistical}. More generally, it would be interesting to think about the best way to communicate uncertainty regarding archetype discovery to an audience of policymakers; perhaps building on the recent results in \cite{andrews2024communicating} and the framework for scientific communication of \cite{andrews2021model}. Third, as argued by \cite{breza2025generalizability}, allowing researchers to admit ignorance and developing tools that identify groups for which more data could be collected can improve the communication of scientific discoveries in social sciences. While we were able to show that, under some assumptions, a dynamic programming algorithm for clustering could be used to construct \emph{basins of ignorance}, it would be desirable to develop more general results for the archetype discovery problem with abstention. A potentially fruitful connection could be based on the work of \cite{DupinNielsen2023PartialKMeans}. Instead of allowing the researcher to pay a cost to admit ignorance about the policy effects associated to some covariate value, one could give the researcher a fixed budget in which ignorance can be admitted for $x\%$ of the observations in the sample. The results in \cite{DupinNielsen2023PartialKMeans} suggest that a problem like this could be computationally tractable, using an algorithm similar in spirit to the one used in this paper.        

\onehalfspacing

\bibliographystyle{ecta}
\bibliography{references.bib}

\doublespacing

\newpage 

\appendix 

\section{Proofs of Main Results} 
\label{app:main_results}
\subsection{Proof of Theorem \ref{thm:archetype discovery solution}} \label{subsec:archetype discovery solution} 

\emph{Preliminaries:} Given a function $\phi: \mathcal{X} \rightarrow \mathbb{R}$ that takes $N>K$ different values $\phi_1<\ldots<\phi_N$, define, for $i=1 \ldots N$, the sets 
\[ G^{\phi}_i \equiv \phi^{-1}(\phi_i) = \{x \in \mathcal{X} \: | \: \phi(x)=\phi_{i}\}.  \]
Note that the set $G^{\phi}_i$ collects the values of $x \in \mathcal{X}$ that satisfy $\phi(x)=\phi_i$. Note that the collection of sets $G^{\phi} \equiv \{G_{1}^{\phi}, \ldots, G^{\phi}_N\}$ forms a partition of $\mathcal{X}$. 

Let $A_K$ be the set of all functions $a:\{1,\ldots, N\} \rightarrow \mathbb{R}$ such that $|a(\{1,\ldots, N\})| \leq K$. Consider then the set of functions 
\[ \bar{\Phi}_K(A_K,G^{\phi}) \equiv \left\{ \bar{\phi} \in \bar{\Phi}_{K} \: \Big | \: \bar{\phi}(x) = \sum_{i=1}^{N} a(i) \mathbf{1}\{x \in G^\phi_i\} \:  \textrm{for some} \: a \in A_K   \right \}.  \]
It can be shown that the set $\bar{\Phi}_K(A_K,G^{\phi})$ coincides with the set of all functions that are measurable with respect to the $\sigma$-algebra generated by the partition $G^{\phi}$ and that take at most $K$ values.  

\emph{Upper Bound:} Note first that  
\begin{equation} \label{eq: upper bound proof Thm 1}
L(\bar{\phi}^*,\phi,p) = \inf_{\bar{\phi} \in \bar{\Phi}_{K}} L(\bar{\phi},\phi,p) \leq \inf_{\bar{\phi} \in \bar{\Phi}_K(A_K,G^{\phi})} L(\bar{\phi},\phi,p) = \inf_{\bar{\phi} \in \bar{\Phi}_K(A_K,G^{\phi})} \sum_{x \in \mathcal{X}} p(x) \left( \phi(x)-\bar{\phi}(x) \right)^2  . 
\end{equation}
We now use the structure of $\bar{\Phi}_K(A_K,G^{\phi})$ to re-write the right-hand side of \eqref{eq: upper bound proof Thm 1}. To this end, fix a function $\bar{\phi} \in \bar{\Phi}_K(A_K,G^{\phi})$ with associated function $a(\cdot)$ and write its image---that is, the set $a(\{1,\ldots,N\})$---as $\{a_1, \ldots, a_{\tilde{K}}\}$, where $\tilde{K} \leq K$. As usual, let $a^{-1}(a_k) = \{ i \in \{1,\ldots, N\} \: | \: a(i) = a_k \}$ denote the indexes of all the partitions $G^{\phi}_i$ over which the function $\bar{\phi}(x)$ takes the value $\alpha_k$. Algebra shows that 
\begin{eqnarray*} 
\sum_{x \in \mathcal{X}} p(x) \left( \phi(x)-\bar{\phi}(x) \right)^2 &=& \sum_{k=1}^{\tilde{K}} \left( \sum_{x \in \cup_{i \in a^{-1}(a_k)} G^{\phi}_i} p(x) \left( \phi(x)-a_k \right)^2 \right) \\
&=& \sum_{k=1}^{\tilde{K}} \left( \sum_{i \in a^{-1}(a_k)} \left(  \sum_{x \in  G^{\phi}_i} p(x) \left( \phi(x)-a_k \right)^2 \right)\right) \\
&=& \sum_{k=1}^{\tilde{K}} \left( \sum_{i \in a^{-1}(a_k)} \left(  \sum_{x \in  G^{\phi}_i} p(x) \left( \phi_i-a_k \right)^2 \right)\right) \\
&=& \sum_{k=1}^{\tilde{K}} \left( \sum_{i \in a^{-1}(a_k)} \left(  \sum_{x \in  G^{\phi}_i} p(x)  \right) \left( \phi_i-a_k \right)^2 \right) \\
&=& \sum_{k=1}^{\tilde{K}} \left( \sum_{i \in a^{-1}(a_k)} p_i \left( \phi_i-a_k \right)^2 \right).
\end{eqnarray*}
Note that the function $a(\cdot)$ plays two roles in the summation above. First, the function $a(\cdot)$ partitions the index set $\{1,\ldots, N\}$ through the inverse images $\{a^{-1}(a_k)\}_{k=1}^{\tilde{K}}$. The interpretation of each set $a^{-1}(a_k)$ is that this set collects the indexes of all partitions $G^{\phi}_i$ over which the function $\bar{\phi}$ takes the value $a_k$. Second, the image of the function $a(\cdot)$ determines the specific values $a_k$. These two roles can be separated as follows. Let $c:\{1,\ldots, N\} \rightarrow \{1,\ldots, K\}$ be a surjective function and let $\mu$ be an arbitrary subset of $\mathbb{R}$ containing at most $K$ different values, and let $\mu_k$ denote its $k$-th value. The right-hand of \eqref{eq: upper bound proof Thm 1} is thus equivalent to solving
\begin{equation} \label{eq: upper bound proof Thm 1 simplified} 
\inf_{c:\{1,\ldots, N\}\rightarrow \{1,\ldots, K\}} \left( \inf_{\mu \textrm{ s.t.} |\mu| \leq K } \sum_{k=1}^{K} \left( \sum_{\{i \:| \: c(i)=k\}} p_i \left( \phi_i-\mu_k \right)^2 \right) \right). 
\end{equation}
But note that given a function $c(\cdot)$, the value 
\[ \mu_k(c) \equiv \frac{\sum_{\{i \: : \: c(i)=k\}} p_i \phi_i }{\sum_{\{i \: : \: c(i)=k\}} p_i} \] 
solves the problem 
\[ \min_{\mu_k \in \mathbb{R}}\sum_{\{i \:| \: c(i)=k\}} p_i \left( \phi_i-\mu_k \right)^2. \]
This means that \eqref{eq: upper bound proof Thm 1 simplified} can be written as
\begin{equation} \label{eq: upper bound proof Thm 1 final}
\inf_{c:\{1,\ldots, N\}\rightarrow \{1,\ldots, K\}} \left( \sum_{k=1}^{K} \left( \sum_{\{i \:| \: c(i)=k\}} p_i \left( \phi_i-\mu_k(c) \right)^2 \right) \right). 
\end{equation}
This shows that the value \eqref{eq: upper bound proof Thm 1 final} is an upper bound for $\inf_{\bar{\phi} \in \bar{\Phi}_{K}} L(\bar{\phi},\phi,p)$ (the value of the archetype discovery problem).

\emph{Lower Bound:} We now want to show that  \eqref{eq: upper bound proof Thm 1 final} is also a lower bound for the value of the archetype discovery problem. The main challenge here is that there are several functions in $\bar{\Phi}_{K}$ that need not be measurable with respect to the $\sigma$-algebra generated by the partition $G^{\phi}$. We will show below that for any function $\bar{\phi} \in \bar{\Phi}_{K}$ there always exists a function $\bar{\phi}' \in \bar{\Phi}_K(A_K,G^{\phi})$ with better payoff.

Take any function \(\bar\phi\in\bar\Phi_K\), and write its image as
\[
\bar\phi(\mathcal{X})=\{\alpha_1,\dots,\alpha_{\tilde{K}}\}.
\]
As usual, let $\bar{\phi}^{-1}(\alpha_k) = \{ x \in \mathcal{X} \: | \: \bar{\phi}(x)=\alpha_k\}$
For each \(i\in\{1,\dots,N\}\) and each
\(k\in\{1,\dots,\tilde{K}\}\), define the probabilities
\[
q_{ik}
\equiv
\sum_{x \in G^{\phi}_i \cap \bar{\phi}^{-1}(\alpha_k)} p(x) = \sum_{x\in G_i^\phi} p(x)\mathbf{1}\{\bar\phi(x)=\alpha_k\}.
\]
Note that $q_{ik}$ represents the mass assigned to the elements inside $G^{\phi}_i$ for which the the function $\bar{\phi}(x)=\alpha_k$. Note also that for every fixed $i \in \{1,\ldots, N\}$
\[
\sum_{k=1}^{\tilde{K}} q_{ik}
=
\sum_{x\in G_i^\phi} p(x)
=
p_i.
\]
Since \(\phi(x)=\sum_{i=1}^{N}\phi_i \mathbf{1}\{x \in G_i^{\phi}\}\) , we can write
\begin{align*}
L(\bar\phi,\phi,p)
&=
\sum_{x\in X} p(x)\bigl(\phi(x)-\bar\phi(x)\bigr)^2 \\
&=
\sum_{i=1}^N \sum_{x\in G_i^\phi} p(x)\bigl(\phi_i-\bar\phi(x)\bigr)^2 \\
&= \sum_{i=1}^N \left( \sum_{k=1}^{\tilde{K}} \sum_{x\in G_i^\phi \cap \bar{\phi}^{-1}(\alpha_k)} p(x)\bigl(\phi_i-\bar\phi(x)\bigr)^2 \right)  \\
&= \sum_{i=1}^N \left( \sum_{k=1}^{\tilde{K}} \sum_{x\in G_i^\phi \cap \bar{\phi}^{-1}(\alpha_k)} p(x)\bigl(\phi_i-\alpha_k\bigr)^2 \right)  \\
&=
\sum_{i=1}^N \sum_{k=1}^{\tilde{K}} q_{ik}(\phi_i-\alpha_k)^2.
\end{align*}

Now, define the function $a:\{1,\ldots, N\} \rightarrow \{1,\ldots, K\}$ to be any function satisfying
\[
a(i)\in \arg\min_{k\in\{1,\dots,\tilde{K}\}} (\phi_i-\alpha_k)^2.
\]

We use the function \(a(\cdot)\) to define a new function \(\bar\phi':X\to\mathbb R\) belonging to the set \(\bar{\Phi}_K(A_K,G^{\phi})\) as follows: 
\[
\bar\phi'(x)
\equiv
\sum_{i=1}^N \alpha_{a(i)}\mathbf{1}\{x\in G_i^\phi\}.
\] 
We now show that the loss of \(\bar\phi'\) is smaller than the loss of \(\bar\phi\). To see this, note that for each \(i\),
\begin{align*}
\sum_{x\in G_i^\phi} p(x)\bigl(\phi(x)-\bar\phi'(x)\bigr)^2
&=
\left( \sum_{x\in G_i^\phi} p(x)\right)\bigl(\phi_i-\alpha_{a(i)}\bigr)^2 \\
&=
p_i(\phi_i-\alpha_{a(i)})^2.
\end{align*}
Because \(a(i)\) minimizes \((\phi_i-\alpha_k)^2\) over \(k\in\{1,\dots,\tilde{K}\}\), we have
\begin{eqnarray*}
p_i(\phi_i-\alpha_{a(i)})^2 &=& \sum_{x \in G^{\phi}_i} p(x)(\phi_i-\alpha_{a(i)})^2 \\
&\le&
\sum_{x \in G^{\phi}_i} p(x) (\phi_i-\bar{\phi}(x))^2. 
\end{eqnarray*}
Therefore, 
\begin{eqnarray*}
L(\bar{\phi}',\phi,p) &=& \sum_{x \in \mathcal{X}}  p(x)\bigl(\phi(x)-\bar\phi'(x)\bigr)^2 \\
&=& \sum_{i=1}^{N} \sum_{x \in G^{\phi}_i} p(x) \bigl(\phi(x)-\bar\phi'(x)\bigr)^2 \\
&\leq& \sum_{i=1}^{N} \sum_{x \in G^{\phi}_i} p(x) \bigl(\phi(x)-\bar\phi(x)\bigr)^2 \\
&=& L(\bar{\phi},\phi,p)
\end{eqnarray*}
Thus, we have shown that for every function \(\bar\phi\in\bar\Phi_K\) there exists a function 
\(\bar\phi' \in \bar{\Phi}_K(A_K,G^{\phi})\) with smaller loss.  

This means that
\[
L(\bar{\phi}^*,\phi,p) = \inf_{\bar{\phi} \in \bar{\Phi}_{K}} L(\bar{\phi},\phi,p) \geq \inf_{\bar{\phi} \in \bar{\Phi}_K(A_K,G^{\phi})} L(\bar{\phi},\phi,p) . 
\]

\emph{Conclusion:} Using the upper and lower bounds above, we have shown that 

\begin{eqnarray*}
L(\bar{\phi}^*,\phi,p) &=&  \inf_{\bar{\phi} \in \bar{\Phi}_{K}} L(\bar{\phi},\phi,p) \\
&=& \inf_{\bar{\phi} \in \bar{\Phi}_K(A_K,G^{\phi})} L(\bar{\phi},\phi,p)\\
&=& \inf_{c:\{1,\ldots, N\}\rightarrow \{1,\ldots, K\}} \left( \sum_{k=1}^{K} \left( \sum_{\{i \:| \: c(i)=k\}} p_i \left( \phi_i-\mu_k(c) \right)^2 \right) \right), 
\end{eqnarray*}
where the last infimum is taken over all $K$-clustering functions. Let $c^*$ be the solution of such a clustering problem. As shown in the construction of the upper bound, the optimal $\bar{\phi}^*$ can be constructed from $c^*$ as follows: 
\[ \bar{\phi}^*(x) = \sum_{i=1}^{N}\mu_{c^*(i)}(c^*) \mathbf{1}\{x \in G^{\phi}_i\}.   \]
The $k$-th archetype set associated to the oracle solution $\phi^*$ equals
\begin{eqnarray*}
\mathcal{A}^*_k &\equiv& \{ x \in \mathcal{X} \: \mid \: \bar{\phi}^*(x) = \bar{\phi}^*_k \}.
\end{eqnarray*}
The function $\bar{\phi}^*$ can be defined alternatively as in the statement of Theorem \ref{thm:archetype discovery solution} by defining $i:\mathcal{X} \rightarrow \{1,\ldots, N\}$ as the function such that assigns each $x$ to the value that $\phi(\cdot)$ takes at such point, this is: $\phi(x) = \phi_{i(x)}$. 

\subsection{Proof of Theorem \ref{thm:posterior_loss_minimization}} \label{subsec:posterior_loss_minimizaation} 
It suffices to show that $d^*(z)$ is a solution to the posterior loss minimization problem
\begin{equation} \label{eq:appendix_posterior_loss_minimization}
\inf_{\bar{\phi} \in \bar{\Phi}_{K}} E_{\phi \sim \pi} \left[ L(\bar{\phi},\phi,p) \: | \: z \right] = \inf_{\bar{\phi} \in \bar{\Phi}_{K}} \left( \sum_{x \in \mathcal{X}} p(x) \mathbb{E}_{\phi \sim \pi} \left[ (\phi(x)-\bar{\phi}(x))^2  \: | \: z\right] \right). 
\end{equation}
Fix a data realization $z \in \mathcal{Z}$, and let
\[ \widehat{\phi}(x) \equiv \mathbb{E}_{\phi \sim \pi}[\phi(x) \: | \: z]  \]
denote the posterior mean of the function $\phi(\cdot)$ given $z \in \mathcal{Z}$. Adding and subtracting the posterior mean function we get that for each $x \in \mathcal{X}$ we have 
\begin{eqnarray*}
\mathbb{E}_{\phi \sim \pi} \left[ (\phi(x)-\bar{\phi}(x))^2  | z\right] &=& \mathbb{E}_{\phi \sim \pi} \left[ (\phi(x)- \widehat{\phi}(x)+\widehat{\phi}(x)- \bar{\phi}(x))^2  | z\right] \\
&=& \mathbb{E}_{\phi \sim \pi} \left[ (\phi(x)- \widehat{\phi}(x))^2 | z \right]+ (\widehat{\phi}(x)- \bar{\phi}(x))^2. 
\end{eqnarray*}
Thus, \eqref{eq:appendix_posterior_loss_minimization} equals
\[\sum_{x \in \mathcal{X}} p(x) \mathbb{E}_{\phi \sim \pi} \left[ (\phi(x)- \widehat{\phi}(x))^2 | z \right]+ \inf_{\bar{\phi} \in \bar{\Phi}_{K}}  \sum_{x \in \mathcal{X}}p(x)(\widehat{\phi}(x)- \bar{\phi}(x))^2. \]
But the minimization problem at the end of the expression above is the same as the oracle archetype discovery problem when $\widehat{\phi}$ is taken as the true $\phi$.

\subsection{Proof of Theorem \ref{thm:epsilon-minimax-archetypes}}
\label{app:proof-epsilon-minimax-theorem}

\begin{proof}
Throughout this section, we use the norm
\[
\|\phi\| \equiv \sup_{x\in\mathcal X}|\phi(x)|.
\]
We establish Theorem \ref{thm:epsilon-minimax-archetypes} using the following high-level conditions which we verify in Appendix~\ref{app:verification-minimax}:
\paragraph{Condition 1 (Uniform consistency of $\widehat{\phi}$).}
For every \(\varepsilon>0\),
\[
\sup_{\theta \in \Theta}
P_{\theta} \!\left(
\|\hat\phi-\phi\|>\varepsilon
\right)\to 0
\qquad\text{as } I \to\infty.
\]

\paragraph{Condition 2 (Boundedness of the loss).}
There exists a constant \(M>0\) such that for every \(\bar\phi\in\bar\Phi_K(B)\) and every \(\phi \in \Theta\),
\[
|L(\bar\phi;\phi,p)|\le M.
\]

\paragraph{Condition 3 (Uniform Lipschitz continuity of the loss).}
There exists a constant \(C>0\) such that for all \(\phi,\phi'\in \Theta\),
\[
\sup_{\bar\phi\in\bar\Phi_K(B)}
\left|
L(\bar\phi;\phi,p)-L(\bar\phi;\phi',p)
\right|
\le
C\|\phi-\phi'\|.
\]

The proof goes as follows. Fix \(\varepsilon>0\). We will show that for all sufficiently large \(I\),
\[
\sup_{\theta \in \Theta}R(d_{\textrm{plug-in}},\theta)\le V(I,\Theta)+\varepsilon.
\]
For fixed $\eta>0$ and $\phi \in \Theta$ let
\[
A_{I}(\phi,\eta):=\{ \hat{\phi} \: | \:   \|\hat\phi-\phi\|\le \eta\},
\]
denote the set of all data realizations for which $\widehat{\phi}$ is at most $\eta$ away from $\phi$. On the event \(A_I(\phi,\eta)\), Condition 3 implies
\[
L(d_{\textrm{plug-in}}(\widehat{\phi});\phi,p)
\le
L( d_{\textrm{plug-in}}(\widehat{\phi});\widehat{\phi}_n,p)+C\eta.
\]
Moreover, because \(d_{\textrm{plug-in}}(\widehat{\phi})\) minimizes \(\bar\phi\mapsto L(\bar\phi;\hat\phi,p)\) for each $\widehat{\phi}$, 
\[
L(d_{\textrm{plug-in}}(\widehat{\phi});\hat\phi,p)
\le
L(d^*(\widehat{\phi});\hat\phi,p),
\]
where $d^*$ is the minimax rule that attains the value $V(\Theta,I)$. Applying Condition 3 again on \(A_I(\phi,\eta)\),
\[
L(d^*(\widehat{\phi});\hat\phi,p)
\le
L(d^*(\widehat{\phi});\phi,p)+C\eta.
\]
Therefore, on $A_I(\phi,\eta)$
\[
L(d_{\textrm{plug-in}}(\widehat{\phi});\phi,p)
\le
L(d^*(\widehat{\phi});\phi,p)+2C\eta.
\]

Using the inequality above and the bound of the loss in Condition 2, we obtain
\begin{align*}
R(d_{\textrm{plug-in}},\theta)
&=
\mathbb{E}_{\theta} \!\left[
L(d_{\textrm{plug-in}}(\widehat{\phi});\phi,p)\mathbf 1\{A_I(\phi,\eta)\}
\right]
+
\mathbb{E}_{\theta}\!\left[
L(d_{\textrm{plug-in}}(\widehat{\phi});\phi,p)\mathbf 1\{A_I(\phi,\eta)^c\}
\right]
\\
&\le
\mathbb{E}_{\theta}\!\left[
L(d^*(\widehat{\phi});\phi,p)\mathbf 1\{A_I(\phi,\eta)\}
\right]
+
2C\eta
+
\left( M \cdot P_{\theta} \!\left(A_I(\phi,\eta)^c\right) \right)
\\
&\le
R(d^*,\theta)+2C\eta+M P_{\theta}\!\left(A_I(\phi,\eta)^c\right).
\end{align*}
Since \(R(d^*,\theta)\le V(I,\Theta)\) for every \(\theta\in\Theta\), it follows that
\[
\sup_{\theta \in \Theta}R(d_{\textrm{plug-in}}(\widehat{\phi}),\theta)
\le
V(I,\Theta) 
+
2C\eta
+
M\sup_{\theta\in \Theta}P_{\theta}\!\left(A_I(\phi,\eta)^c\right).
\]

Now choose \(\eta=\varepsilon/(4C)\). Then \(2C\eta=\varepsilon/2\). By the uniform consistency of \(\hat\phi\), for this fixed \(\eta\) there exists \(I(\varepsilon)\) such that for all \(I\ge I(\varepsilon)\),
\[
M\sup_{\theta \in \Theta}P_{\theta}\!\left(A_I(\phi,\eta)^c\right)\le \varepsilon/2.
\]
Hence, for all \(I\ge I(\varepsilon)\),
\[
\sup_{\theta \in \Theta}R(d_{\textrm{plug-in}}(\widehat{\phi}),\theta)\le V(I,\Theta)+\varepsilon.
\]
This proves that \(d_{\textrm{plug-in}}(\widehat{\phi})\) is \(\varepsilon\)-minimax for all sufficiently large \(I\).
 
\end{proof}

\subsection{Proof of Theorem \ref{thm:minimax-rate}}
\label{app:proof-rate-epsilon-minimax-theorem}

\begin{proof}
 
Fix $\theta=(\phi,P)\in\Theta(B)$. Let $a_{\mathrm{oracle}}$ denote the oracle solution to the archetype discovery problem. By definition,
\[
L(a_{\mathrm{oracle}};\phi,p) = \inf_{\bar{\phi} \in \bar{\Phi}_{K}(B)} L(\bar{\phi};\phi,p).
\]
Since $d_{\mathrm{plug\text{-}in}}(\hat\phi)$ minimizes
$\bar\phi \mapsto L(\bar\phi;\hat\phi,p)$ over $\bar\Phi_K(B)$ pointwise,
and since $a_{\mathrm{oracle}}\in\bar\Phi_K(B)$, we have
\[
L(d_{\mathrm{plug\text{-}in}}(\hat\phi);\hat\phi,p)
\;\le\;
L(a_{\mathrm{oracle}};\hat\phi,p)
\]
for every data realization. Using twice the Lipschitz Condition 3 used in the proof of Theorem \ref{thm:epsilon-minimax-archetypes}, gives
\begin{align*}
\mathbb{E}_{\widehat{\phi} \sim (\phi,P)} \left[ \mathcal{L}\left( d_{\mathrm{plug\text{-}in}}(\hat\phi);\phi,p \right) \right]
&=
\mathbb{E}_{\widehat{\phi} \sim (\phi,P)}  \!\left[
L(d_{\mathrm{plug\text{-}in}}(\hat\phi);\phi,p)
-
L(a_{\mathrm{oracle}};\phi,p)
\right] \\
&\le
\mathbb{E}_{\widehat{\phi} \sim (\phi,P)} \!\left[
L(d_{\mathrm{plug\text{-}in}}(\hat\phi);\phi,p)
-
L(d_{\mathrm{plug\text{-}in}}(\hat\phi);\hat\phi,p)
\right] \\
&\quad+
\mathbb{E}_{\widehat{\phi} \sim (\phi,P)} \!\left[
L(a_{\mathrm{oracle}};\hat\phi,p)
-
L(a_{\mathrm{oracle}};\phi,p)
\right] \\
&\le
8B \,
\mathbb{E}_{\widehat{\phi} \sim (\phi,P)} \!\left[\|\hat\phi-\phi\|_\infty\right].
\end{align*}
By the maximal inequality for sub-Gaussian random variables (Theorem~2.5 of \citet{boucheron2013concentration}), if for each \(x\in\mathcal X\),
\[
Z_x := \hat{\phi}(x)-\phi(x)
\]
is sub-Gaussian with variance proxy \(\bar{\sigma}^2/I\), then
\[
\mathbb{E}\bigl[\|\hat{\phi}-\phi\|_{\infty}\bigr]
=
\mathbb{E}\Bigl[\max_{x\in\mathcal X}|Z_x|\Bigr]
\le
\bar{\sigma}\sqrt{\frac{2\log(2|\mathcal X|)}{I}}.
\]
Therefore,
\[
\sup_{\theta\in\Theta(B)}
\mathbb{E}_{\hat{\phi}}\bigl[\|\hat{\phi}-\phi\|_{\infty}\bigr]
\le
\bar{\sigma}\sqrt{\frac{2\log(2|\mathcal X|)}{I}}.
\]
Taking the supremum over $\theta\in\Theta(B)$ and using the subgaussian
tail bound inequality,
\[
\sup_{\theta \in \Theta(B)}
\mathbb{E}_{\hat\phi}\!\left[\|\hat\phi-\phi\|_\infty\right]
\;\le\;
\bar\sigma\sqrt{\frac{2 \log(2|\mathcal{X}|)}{I}}.
\]
Therefore,
\[
\adjustlimits \inf_{d}\sup_{\theta \in \Theta(B)}
\mathbb{E}_{\widehat{\phi} \sim (\phi,P)} \left[ \mathcal{L}\left( d( \widehat{\phi});\phi,p \right) \right]
\;\le\; 
\sup_{\theta \in \Theta(B)}
\mathbb{E}_{\widehat{\phi} \sim (\phi,P)} \left[ \mathcal{L}\left( d_{\textrm{plug-in}}( \widehat{\phi});\phi,p \right) \right] 
\leq 
8B \bar\sigma\, \sqrt{\frac{2 \log(2|\mathcal{X}|)}{I}},
\]
where $\bar{\sigma}$ is the largest value of $\sigma_{x}$. 

\end{proof}

\subsection{Proof of Theorem \ref{thm:thm_oracle_abstention_simple}} \label{subsec:proof_oracle_abstention}
For $i=1 \ldots N$, define the sets 
\[ G^{\phi}_i \equiv \phi^{-1}(\phi_i) = \{x \in \mathcal{X} \: | \: \phi(x)=\phi_{i}\}.  \]
The loss for the archetype discovery problem with abstention can be re-written as
\begin{eqnarray*}\label{eq:} 
L(\bar{\phi},\pi;\phi,p) &\equiv& \sum_{x\in\mathcal{X}} p(x) \Big[ \pi(x) \big(\phi(x) - \bar\phi(x)\big)^2 + (1-\pi(x))\, \sigma^2 \Big] \\
&=&  \sigma^2 + \sum_{x\in\mathcal{X}} p(x) \Big[ \pi(x) \left( \big(\phi(x) - \bar\phi(x)\big)^2 - \sigma^2 \right) \Big]  \\
&=& \sigma^2 + \sum_{i=1}^{N} \sum_{x\in G^{\phi}_i} p(x) \Big[ \pi(x) \left( \big(\phi_i - \bar\phi(x)\big)^2 - \sigma^2 \right) \Big].
\end{eqnarray*}
For any $\bar{\phi} \in \bar{\Phi}^*_K(A_K,G^{\phi})$, we can write
\[\bar{\phi}(x) = \sum_{i=1}^{N} a(i) \mathbf{1}\{x \in G^\phi_i\} \]
for some function $a:\{1,\ldots, N\} \rightarrow \mathbb{R}$ such that $|a(\{1,\ldots, N\})| = K$. Consequently,  
\[L(\bar{\phi},\pi;\phi,p) = \sigma^2 + \sum_{i=1}^{N} \sum_{x\in G^{\phi}_i} p(x) \Big[ \pi(x) \left( \big(\phi_i - a(i)\big)^2 - \sigma^2 \right) \Big]. \]
In fact, let the $K$-values in $a\{1,\ldots, N\}$ be denoted as $\{\bar{\phi}_1, \ldots, \bar{\phi}_{K}\}$. Since $\pi \in \Pi_{K}(\bar{\phi})$, the function $\pi(x)$ can be written as 
\[ \pi(x)= \sum_{k=1}^{K} \pi_k \mathbf{1}\{ x \in \bar{\phi}^{-1}(\bar{\phi}_k)  \}, \]
with $\pi_{k} \in \{0,1\}$ for every $k=1,\ldots, K$. Therefore, recalling the definition $p_i \equiv \sum_{x\in G^{\phi}_i} p(x)$ we have
\begin{eqnarray*}
L(\bar{\phi},\pi;\phi,p) &=& \sigma^2 + \sum_{k=1}^{K} \sum_{\{i \: | \: a(i)=\bar{\phi}_k \}} \left( \sum_{x\in G^{\phi}_i} p(x) \Big[ \pi(x) \left( \big(\phi_i - \bar{\phi}_k\big)^2 - \sigma^2 \right) \Big] \right) \\
&=&  \sigma^2 + \sum_{k=1}^{K} \sum_{\{i \: | \: a(i)=\bar{\phi}_k \}} \left( \sum_{x\in G^{\phi}_i} p(x) \Big[ \pi_k \left( \big(\phi_i - \bar{\phi}_k\big)^2 - \sigma^2 \right) \Big] \right)\\
&=& \sigma^2 + \sum_{k=1}^{K} \pi_k \sum_{\{i \: | \: a(i)=\bar{\phi}_k \}} \left( p_i \Big[   \big(\phi_i - \bar{\phi}_k\big)^2 - \sigma^2  \Big] \right).
\end{eqnarray*}
Note that we can first fix $\bar{\phi}$ and minimize over all $\pi \in \Pi_{K}(\bar{\phi})$. This can be done by setting each $\pi_k$ with $k=1,\ldots, K$ to $\pi_k= \pi_k^*(\bar{\phi})$, where
\[ \pi^*_k(\bar{\phi}) \equiv \mathbf{1}\left \{ \sum_{\{i \: | \: a(i)=\bar{\phi}_k \}} \left( p_i \Big[   \big(\phi_i - \bar{\phi}_k\big)^2 - \sigma^2  \Big] \right)  \leq 0  \right \}. \]
This means that the profiled loss (after optimizing over $\pi$) can be written as 
\[ L^*(\bar{\phi};\phi,p) = \sigma^2 + \sum_{k=1}^{K} \min \left\{ \sum_{\{i \: | \: a(i)=\bar{\phi}_k \}}  p_i \Big[   \big(\phi_i - \bar{\phi}_k\big)^2 - \sigma^2  \Big],0 \right\}.  \]
It remains to optimize this loss over the function $\bar{\phi}$. Note first that that the function $a$ assigns each element $i \in \{1,\ldots, N\}$ to a value in $\{ \bar{\phi}_1, \ldots, \bar{\phi}_{K}\}$. Let $\textrm{Index}$ be the function that takes an element in the set $\{ \bar{\phi}_1, \ldots, \bar{\phi}_{K}\}$ an retrieves its index. We can then use $a$ to define a clustering function $c: \{1, \ldots, N\} \rightarrow \{1,\ldots, K\}$ as $c(i)=\textrm{Index}(a(i))$. The function $L^*(\bar{\phi};\phi,p)$ can then be written in terms of the clustering function, $c$, and the values $\bar{\phi}_1 \ldots, \bar{\phi}_{K}$. In fact, minimizing $L^*(\bar{\phi};\phi,p)$ over $\bar{\phi} \in \bar{\Phi}^*_K(A_K,G^{\phi})$ is equivalent to minimizing the function 
\[ L^*(c, \bar{\phi}_1,\ldots, \bar{\phi}_{K};\phi,p) = \sum_{k=1}^{K} \min \left\{ \sum_{\{i \: | \: c(i)=k \}}  p_i \Big[   \big(\phi_i - \bar{\phi}_k\big)^2 - \sigma^2  \Big],0 \right\} \]
over clustering functions, $c$, and values $\bar{\phi}_1 \ldots, \bar{\phi}_{K}$. The same argument we used in Theorem \ref{thm:archetype discovery solution} shows that $\bar{\phi}_{k}$ can be chosen as the cluster centers $\mu_{k}(c)$ defined in Theorem \ref{thm:archetype discovery solution}. This means that an oracle solution to the archetype discovery problem with abstention in \eqref{eq:archetype discovery problem abstention} can be found by solving 
\[ \min_{c:\{1,\ldots, N\} \rightarrow \{1,\ldots, K\}} \sum_{k=1}^{K} \min \left\{ \sum_{\{i \: | \: c(i)=k \}}  p_i \Big[   \big(\phi_i - \mu_{k}(c)\big)^2 - \sigma^2  \Big],0 \right\}.   \]
Let $c^*$ denotes the minimizer of this function. Then, the oracle solution to the archetype discovery problem with abstention in \eqref{eq:archetype discovery problem abstention} can be obtained by setting
$\bar{\phi}(x) = \mu_{c^*(i(x))}(c^*)$.

\newpage

\section{Additional Derivations} 
\label{app:support_materials}

\subsection{Verification of the conditions for asymptotic minimax optimality}
\label{app:verification-minimax}

In this appendix, we verify the three high-level conditions used to show that the plug-in rule is \(\varepsilon\)-minimax. 

We verify Conditions 1-3 under the assumptions of Theorem \ref{thm:epsilon-minimax-archetypes}. Let
\[
\bar\sigma:=\sup_{x\in\mathcal X}\sigma_x<\infty.
\]

\noindent\textbf{Verification of Condition 1.}
Under the statistical model in \eqref{eq:estimator},
\[
\hat\phi(x)-\phi(x)=\frac{\sigma_x}{\sqrt I} u_x,
\qquad
\{u_x\}_{x \in \mathcal{X}}\sim P,
\]
Fix $\eta>0$, and let
\[
\bar\sigma:=\sup_{x\in\mathcal X}\sigma_x<\infty.
\]
For any $x\in\mathcal X$ and any $r>0$, Chernoff's bound and the standard equivalence between the
moment-generating-function and tail formulations of subgaussianity (see, e.g.,
\cite[Proposition 2.5.2]{Vershynin2018}) imply
\[
P_\theta(u_x>r)
\le
\inf_{\lambda>0} e^{-\lambda r}\mathbb{E}_\theta[e^{\lambda u_x}]
\le
\inf_{\lambda>0} e^{-\lambda r+\lambda^2/2}
=
e^{-r^2/2}.
\]
Applying the same argument to $-u_x$ yields
\[
P_\theta(|u_x|>r)\le 2e^{-r^2/2}.
\]
Therefore, for every $x\in\mathcal X$,
\[
P_\theta\!\left(|\hat\phi(x)-\phi(x)|>\eta\right)
=
P_\theta\!\left(|u_x|>\frac{\sqrt I\,\eta}{\sigma_x}\right)
\le
2\exp\!\left(-\frac{I\eta^2}{2\sigma_x^2}\right)
\le
2\exp\!\left(-\frac{I\eta^2}{2\bar\sigma^2}\right).
\]
Using the union bound,
\begin{align*}
P_\theta\!\left(\|\hat\phi-\phi\|>\eta\right)
&=
P_\theta\!\left(\sup_{x\in\mathcal X}|\hat\phi(x)-\phi(x)|>\eta\right) \\
&\le
\sum_{x\in\mathcal X}
P_\theta\!\left(|\hat\phi(x)-\phi(x)|>\eta\right) \\
&\le
2|\mathcal X|
\exp\!\left(-\frac{I\eta^2}{2\bar\sigma^2}\right).
\end{align*}
Since the bound does not depend on $\theta=(\phi,P)$, it follows that
\[
\sup_{\theta\in\Theta}
P_\theta\!\left(\|\hat\phi-\phi\|>\eta\right)
\le
2|\mathcal X|
\exp\!\left(-\frac{I\eta^2}{2\bar\sigma^2}\right).
\]
Hence, if $\log |\mathcal X|/I\to 0$, then for every fixed
$\eta>0$,
\[
\sup_{\theta\in\Theta}
P_\theta\!\left(\|\hat\phi-\phi\|>\eta\right)\to 0
\]
which proves Condition~1.

\noindent\textbf{Verification of Condition 2.}
Under the boundedness restrictions,
\[
|\phi(x)-\bar\phi(x)|\le 2B
\qquad\forall x\in\mathcal X.
\]
Therefore,
\[
L(\bar\phi;\phi,p)
=
\sum_{x\in\mathcal X}p(x)\bigl(\phi(x)-\bar\phi(x)\bigr)^2
\le
\sum_{x\in\mathcal X}p(x)(2B)^2
=
4B^2.
\]
Since the loss is nonnegative, it follows that
\[
|L(\bar\phi;\phi,p)|\le 4B^2
\qquad
\forall \bar\phi\in\bar\Phi_K,\ \phi\in\Phi.
\]
Thus Condition 2 holds with \(M \equiv 4B^2\).

\medskip

\noindent\textbf{Verification of Condition 3.}
Fix \(\bar\phi\in\bar\Phi_K(B)\), and let \(\phi,\phi'\in\Theta\). Then
\[
L(\bar\phi;\phi,p)-L(\bar\phi;\phi',p)
=
\sum_{x\in\mathcal X}
p(x)\left[
(\phi(x)-\bar\phi(x))^2-(\phi'(x)-\bar\phi(x))^2
\right].
\]
Expanding and regrouping terms gives
\[
L(\bar\phi;\phi,p)-L(\bar\phi;\phi',p)
=
\sum_{x\in\mathcal X}
p(x)\bigl(\phi(x)-\phi'(x)\bigr)
\bigl(\phi(x)+\phi'(x)-2\bar\phi(x)\bigr).
\]
Therefore,
\begin{align*}
\left|L(\bar\phi;\phi,p)-L(\bar\phi;\phi',p)\right|
&\le
\sum_{x\in\mathcal X}
p(x)\,
|\phi(x)-\phi'(x)|\,
|\phi(x)+\phi'(x)-2\bar\phi(x)|.
\end{align*}
Because
\[
|\phi(x)|\le B,\qquad
|\phi'(x)|\le B,\qquad
|\bar\phi(x)|\le B,
\]
we have
\[
|\phi(x)+\phi'(x)-2\bar\phi(x)|
\le
|\phi(x)|+|\phi'(x)|+2|\bar\phi(x)|
\le 4B.
\]
Hence
\begin{align*}
\left|L(\bar\phi;\phi,p)-L(\bar\phi;\phi',p)\right|
&\le
4B\sum_{x\in\mathcal X}p(x)|\phi(x)-\phi'(x)|
\\
&\le
4B\sup_{x\in\mathcal X}|\phi(x)-\phi'(x)|
\\
&=
4B\,\|\phi-\phi'\|.
\end{align*}
Taking the supremum over \(\bar\phi\in\bar\Phi_K(B)\) yields
\[
\sup_{\bar\phi\in\bar\Phi_K}
\left|L(\bar\phi;\phi,p)-L(\bar\phi;\phi',p)\right|
\le
4B\,\|\phi-\phi'\|,
\]
so Condition 3 holds with \(C \equiv 4B\).

\subsection{Additional Illustrative Example} \label{subsec:dp-illustration-Example2} 

The Atlantic Causal Inference Conference (ACIC) 2016 data are a semi-synthetic benchmark constructed by \citet{Dorie2019}. The design combines real covariates drawn from a study of maternal and infant health called the Collaborative Perinatal Project (CPP) (see, for example \cite{niswander1972} and \cite{klebanoff2009}) with simulated treatment assignments and potential outcomes designed to emulate observational causal inference settings.\footnote{The ACIC 2016 competition data are available through the \texttt{aciccomp2016} package and the associated repository. The competition design and the data-generating processes are described in \cite{Dorie2019}}. In the ACIC 2016 data, the covariates are selected as plausible confounders for a hypothetical study of the effect of infant birth weight on a child's Intelligence Quotient (IQ) \citep{Dorie2019}. In the public release, these variables are anonymized and reported as $x_1,\dots,x_{58}$, but Appendix A.1, Table 5 of \citet{Naghi2021} provides a table of the labels associated to all of the CPP covariates. The retained variables span maternal demographics and socioeconomic characteristics (for example, maternal age, race, education, and family income), pregnancy and birth conditions (for example, gestation at delivery, placental weight, cord length, and Apgar scores), and infant health measures (for example, bilirubin, hematocrit, and hemoglobin). The original ACIC 2016 data contained 77 data-generating processes. We focus on the data generating process \#15 from Appendix A.1, Table 2 of \citet{Dorie2019}. This setting is supposed to feature high treatment-effect heterogeneity. The released data provides the true outcomes of interest $(\mu_0(x),\mu_1(x))$ for each unit, which we use to construct the policy effects of interest $\phi(x):=\mu_1(x)-\mu_0(x)$.

In this example \(|\mathcal X|=4{,}802\), and the number of different heterogeneous treatment effects is \(N=4{,}704\). We choose to summarize \(\phi\) using \(K=10\). The exact dynamic program yields \(L_{\mathrm{oracle}}=0.081\) and has a runtime of \(0.1485\) seconds.  As explained before, the ACIC 2016 data is a hypothetical study of the effect of birth weight on IQ, thus we think of the outcomes encoded in $\phi$ as being expressed in terms of the gains in IQ due to ``high'' birth weight. The minimum value of $\phi$ in this example is .599 and the largest value is 21.267. The 5\% and 95\% quantiles are 1.250 and 6.843, respectively. The average value of $\phi$ is 2.928. The variance is 4.33. These numbers suggest that there is indeed a fair amount of heterogeneity in $\phi$. 

Figure \ref{fig:acic-kmeans-sorted-variance} shows that the $K=10$ clusters obtained by the dynamic programming algorithm successfully cluster observations into more homogeneous groups. It is helpful to compare the clusters which the smallest and largest outcomes, which exhibit markedly different properties. The cluster with the largest outcomes has the largest within-cluster variance, which is about 3 and almost has the same magnitude as the original variance. The within-cluster mean for this group---the function $\bar{\phi}(x)$ for any $x$ in the last cluster---is 16.87. The variance of this cluster is about 3.02. The size of this cluster is fairly small: it only contains .3\% of the data. The values of $\phi$ in the first (left-most) cluster range from .599 to 1.586. The within-cluster mean for this group is 1.34. The variance of the first cluster is about 0.03 (less than 1\% of the original variance) and it contains about 30\% of the data.

\begin{figure}[!htbp]
    \centering
    \includegraphics[scale=.7, width=\textwidth]{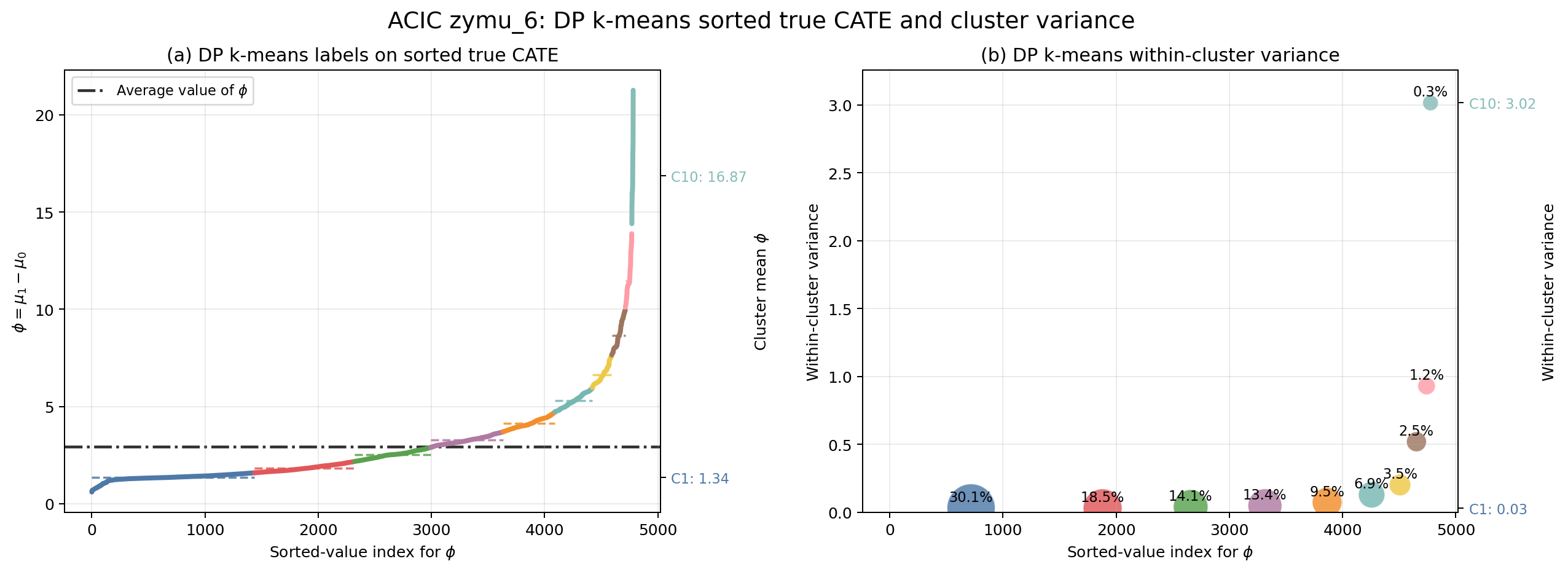}
    \caption{Solution to the archetype discovery problem for the ACIC 2016 data.}
    \label{fig:acic-kmeans-sorted-variance}
\end{figure}

\end{document}